\documentclass[preprint,amsmath,amssymb,aip,jcp,nobibnotes,longbibliography,floatfix]{revtex4-1}  

\usepackage{inputenc}
\usepackage{amsmath}
\usepackage{amssymb, amsfonts}
\usepackage{mathtools}
\usepackage{graphicx,epstopdf,bm}
\usepackage{tabularx}
\usepackage[colorlinks,bookmarks,urlcolor=blue,citecolor=blue,linkcolor=blue]{hyperref}
\usepackage[T1]{fontenc}
\usepackage{bbm}
\usepackage{mathtools}
\usepackage{cases}
\usepackage[shortlabels]{enumitem}
\usepackage{algorithm}
\usepackage{algpseudocode}

\usepackage{bbm}

\renewcommand{\vec}[1]{\boldsymbol{#1}}
\newcommand{\vecs}[1]{\boldsymbol{#1}}
\newcommand{\paren}[1]{\left(#1\right)}
\newcommand{\brac}[1]{\left[#1\right]}
\newcommand{\E}{\mathbb{E}}
\newcommand{\avg}[1]{\E[#1]}

\newcommand{\D}[2]{\frac{d#1}{d#2}}
\newcommand{\PD}[2]{\frac{\partial#1}{\partial#2}}

\newcommand{\abs}[1]{\left|#1\right|}
\newcommand{\norm}[1]{\Vert#1\Vert}

\DeclareMathOperator{\prob}{Pr}

\newcommand{\DA}{D^{\textrm{A}}}
\newcommand{\DB}{D^{\textrm{B}}}
\newcommand{\DC}{D^{\textrm{C}}}

\newcommand{\veta}{\vecs{\eta}}

\newcommand{\vx}{\vec{x}}
\newcommand{\vy}{\vec{y}}
\newcommand{\vz}{\vec{z}}
\newcommand{\vX}{\vec{X}}

\newcommand{\vv}{\vec{v}}

\newcommand{\vw}{\vec{w}}

\newcommand{\vO}{\vec{0}}

\newcommand{\rb}{r_{\textrm{b}}}

\newcommand{\ind}{\mathbbm{1}}

\def\R{\mathbb{R}}

\newcommand{\pb}{p_{\textrm{b}}}
\newcommand{\pbO}{p_{\textrm{b},0}}

\newcommand{\equil}[1]{\bar{#1}}
\newcommand{\peq}{\equil{p}}
\newcommand{\pbeq}{\equil{p}_{\textrm{b}}}

\renewcommand{\L}{\mathcal{L}}
\newcommand{\Lb}{\mathcal{L}_{\textrm{b}}}

\newcommand{\Rset}{\mathcal{R}}
\newcommand{\dRset}{\partial \mathcal{R}}

\newcommand{\Dmat}{\mathcal{D}}
\newcommand{\Dbmat}{\mathcal{D}_{\textrm{b}}}
\newcommand{\Kd}{K_{\textrm{d}}}
\newcommand{\Id}{\mathcal{I}_d}

\newcommand{\dB}{\partial B}

\newcommand{\OSqD}{\Omega^{2}_{\text{free}}}

\renewcommand{\epsilon}{\varepsilon}
\renewcommand{\rb}{\epsilon}

\newcommand{\kp}{\kappa^{+}}
\newcommand{\km}{\kappa^{-}}

\newcommand{\mup}{\beta^{+}}
\newcommand{\mum}{\beta^{-}}

\newcommand{\kphat}{\hat{\kappa}^{+}}

\newcommand{\Omegaf}{\Omega_{\textrm{free}}}

\newcommand{\vR}{\vec{r}}

\usepackage{amsthm}
\theoremstyle{plain}
\newtheorem{theorem}{Theorem}[section]

\begin{document}

\title{Detailed Balance for Particle Models of Reversible Reactions in Bounded Domains}

\author{Ying Zhang}
\email{yingzhang@brandeis.edu}
\affiliation{Department of Mathematics, Brandeis University}

\author{Samuel A. Isaacson}
\email{isaacson@math.bu.edu}
\affiliation{Department of Mathematics and Statistics, Boston University}

\numberwithin{equation}{section}
 
\begin{abstract}
In particle-based stochastic reaction-diffusion models, 
reaction rate and placement kernels are used to decide the probability per time a reaction can occur between reactant particles, and to decide where product particles should be placed. When choosing kernels to use in reversible reactions, a key constraint is to ensure that detailed balance of spatial reaction-fluxes holds at all points at equilibrium. In this work we formulate a general partial-integral differential equation model that encompasses several of the commonly used contact reactivity (e.g. Smoluchowski-Collins-Kimball) and volume reactivity (e.g. Doi) particle models. From these equations we derive a detailed balance condition for the reversible $\textrm{A} + \textrm{B} \leftrightarrows \textrm{C}$ reaction. In bounded domains with no-flux boundary conditions, when choosing unbinding kernels consistent with several commonly used binding kernels, we show that preserving detailed balance of spatial reaction-fluxes at all points requires spatially varying unbinding rate functions near the domain boundary. Brownian Dynamics simulation algorithms can realize such varying rates through ignoring domain boundaries during unbinding and rejecting unbinding events that result in product particles being placed outside the domain. 
\end{abstract} 
  
\maketitle 

\section{Introduction}
Particle-based stochastic reaction-diffusion (PBSRD) models are a common approach for studying biochemical systems where stochasticity in both the diffusive motion of particles and reactive interactions between particles are important.  They have been used in studying a variety of spatially-distributed cellular and biological systems. Examples include how molecular reach can control the efficacy of T-cell activation within the cell membrane~\cite{Zhang2019}, how noise can influence the response of spatially-distributed signaling pathways~\cite{WoldeEgfrdPNAS2010}, and how the dynamics and formation of protein clusters are tuned to balance cluster size and protein mobility~\cite{Ullrich2015cy}.

In PBSRD models, the state of a chemical system is given by the collective chemical states and positions of all particles. PBSRD models represent an intermediate physical scale between computationally expensive microscopic all-atom molecular dynamics models~\cite{ShawAntonMS2009}, and macroscopic mean-field chemical kinetics models in which biochemical systems are described through a system of reaction-diffusion partial differential equations (PDEs) for the spatially-dependent concentrations of chemical species. 

In studying spatially-distributed biological and chemical reaction processes, there are several classes of PBSRD models that have been used in applications. In this work we focus on models that treat molecules as point-particles moving by Brownian motion, but note that these models can be generalized to systems where particles have physical sizes~\cite{KleinSchwarz2014,SmoldynVolEx2017} and/or move by drift-diffusion~\cite{FrohnerNoe2018}. We consider two main classes of PBSRD models, distinguished in how they model bimolecular reactions such as $\textrm{A}+\textrm{B}\to\textrm{C}$. The first class are contact-reactivity (CR) models, where two molecules of species $\textrm{A}$ and $\textrm{B}$ may react upon reaching a fixed separation. This includes the popular Smoluchowski-Collins-Kimball (SCK) model, in which a pair of reactant particles have a probability of successfully reacting, or reflecting apart, once they reach some specified reaction-radius, $\rb$~\cite{CollinsKimballPartialAdsorp}. Note that while the SCK model can account for molecular sizes in bimolecular reaction processes via appropriately chosen reaction-radii, in many applications molecule densities are assumed sufficiently dilute that molecules are otherwise treated as point particles, offering improved computational performance~\cite{AndrewsBrayPhysBio2004,SmoldynVolEx2017}. It is this latter form of the SCK model we consider in this work.

The second general class of PBSRD models are volume reactivity (VR) models, in which $\textrm{A}$ and $\textrm{B}$ particles react with probabilities per time based on their current positions. Perhaps the most common VR model is the Doi~\cite{DoiSecondQuantA,DoiSecondQuantB,PrustelMeierS2014} or $\lambda$-$\rho$~\cite{ErbanChapman2009} model (which Doi attributes to~\cite{TeramotoDoiModel1967}), in which two reactants react with a fixed probability per time when within a reaction-radius of each other. As we show in the next section, in both the CR and VR models these reaction choices can be encoded through a specified reaction kernel, which determines the probability per time individual pairs of reactants may react based on their positions, and the probability density reaction products are placed at given positions.

Both the CR and VR models are the basis for a number of widely-used particle-based simulation packages. These include Brownian Dynamics (BD) simulators such Smoldyn~\cite{AndrewsBrayPhysBio2004}, which was originally designed to support the Smoluchowski model but now supports volume exclusion~\cite{SmoldynVolEx2017} and several different PBSRD models; and ReaDDy, a Doi-model based simulator~\cite{Schoneberg2013ek,NoeReaddy22019}. A variety of timestep-based BD type methods \cite{AndrewsBrayPhysBio2004,Schoneberg2013ek,Donevetal2018,MorellitenWold2008}, spatially-discrete continuous-time jump process methods~\cite{IsaacsonZhang17,IsaacsonCRDME2013}, exact propagation methods~\cite{WoldeEgfrdPNAS2010,DonevJCP2010}, and efficient coarse-grained and multiscale simulation methods~\cite{FranzEtAl2013,FleggEtAl2015,HarrisonYates2016,SmithYates2021,BlackwellKoh2011,BlackwellKoh2012,PrustelMeierS2021} have been proposed for simulating various CR and VR PBSRD models. We note that while we later discuss reaction product placement models near boundaries based on models developed for Smoldyn, in its default mode Smoldyn simulates Smoluchowski dynamics~\cite{SmoluchowskiDiffLimRx}, which is not a special case of either the VR or CR models (but does arise as a limit of both~\cite{AgbanusiIsaacsonDoi,KeizerJPhysChem82}). Smoldyn also chooses reaction parameters not as independent model features, but in a timestep dependent manner~\cite{AndrewsBrayPhysBio2004}.

In this work we investigate a basic equilibrium property of the CR and VR PBSRD models; whether they preserve detailed balance of (pointwise) spatial reaction fluxes at equilibrium for reversible reactions. We study the  $\textrm{A}+\textrm{B} \leftrightarrows \textrm{C}$ reaction, in the simplified case of a system with just one particle of species $\textrm{A}$ at $\vx$ and one particle of species $\textrm{B}$ at $\vy$, or one particle of species $\textrm{C}$ at $\vz$. For this system, detailed balance of (pointwise) spatial reaction fluxes is the statement that at equilibrium the following are equal
\begin{enumerate}
  \item The probability density per time the system is in the unbound state and the $\textrm{A}$ particle at $\vx$ reacts with the $\textrm{B}$ particle at $\vy$ to produce a $\textrm{C}$ particle at $\vz$.
  \item The probability density per time the system is in the bound state and the $\textrm{C}$ particle at $\vz$ dissociates into an $\textrm{A}$ particle at $\vx$ and a $\textrm{B}$ particle at $\vy$.
\end{enumerate}
\phantom{word}

For reversible reactions, microscopic reversibility of quantum mechanical systems can, via systematic approximations, be argued to result in detailed balance of reaction fluxes at equilibrium for macroscopic well-mixed reaction systems~\cite{MolecRxDynHenriksen,VanKampenDBI57}. Though we are not aware of any rigorous derivations, we similarly expect that microscopic reversibility also implies detailed balance of forward and backward (pointwise) spatial reaction fluxes in PBSRD models at equilibrium. From a statistical mechanical perspective, it has been postulated that reversible chemical reactions should not alter the state of thermodynamic equilibrium, so that diffusing particles are well-mixed at equilibrium and detailed balance of (pointwise) spatial reaction fluxes holds~\cite{Donevetal2018}.

Preserving detailed balance of physical processes has been shown to be important in modeling transport processes within cells~\cite{ElstonPeskinJTB2003}, and choices of reaction kernels that violate detailed balance have been shown to cause convergence to non-equilibrium steady states for closed particle systems~\cite{FrohnerNoe2018}. In addition, preserving detailed balance of (pointwise) spatial reversible reaction fluxes has also been a key design consideration in several recent numerical methods and simulation packages~\cite{MorellitenWold2008,FrohnerNoe2018,Donevetal2018,IsaacsonZhang17}. Other methods may not rigorously preserve the detailed balance of (pointwise) spatial reaction-fluxes, but have been designed to still accurately capture equilibrium properties such as equilibrium and dissociation constants at the population level~\cite{AndrewsBrayPhysBio2004,AndrewsRebinding2005}. Note, in most of these methods detailed balance was only presented for PBSRD models in periodic or unbounded domains~\cite{MorellitenWold2008,FrohnerNoe2018,Donevetal2018}. 

In many contexts, for example modeling cellular processes, PBSRD models are used in closed and bounded domains with reflecting boundary conditions, where detailed balance of (pointwise) spatial reaction fluxes would also be expected to hold at equilibrium. This raises the question of whether previously proposed reaction kernels ensure detailed balance in such domains, or whether modifications are needed to account for more general geometries and/or reflecting boundaries. In this work, we derive and explore a pointwise detailed balance condition for the CR and VR PBSRD models in bounded domains with no-flux reflecting boundary conditions. 

Unless stated otherwise, in the remainder ``detailed balance'' will refer to the microscopic balance of (pointwise) spatial reaction fluxes at equilibrium. Similarly, ``reaction rates'' will refer to the microscopic PBSRD parameters determining the probability per time a reaction can occur given one or more particles in an appropriate configuration to possibly react.

We begin in the next section by formulating a generalized equation that can be used to represent either the CR or VR model for a pair of $\textrm{A}$ and $\textrm{B}$ molecules undergoing the reversible $\textrm{A} + \textrm{B} \leftrightarrows \textrm{C}$ reaction. We then explain in Section~\ref{sect:SSDBandEqSoln} how requiring detailed balance to hold at equilibrium determines the general solution to this model, and requires that the association and dissociation kernels are proportional. In Section~\ref{sect:DB_for_common_models} we consider several popular choices for association and dissociation kernels in the VR and CR models, investigating whether they allow detailed balance to hold in bounded domains with reflecting boundary conditions. Next, in Section~\ref{sect:dbreject} we demonstrate a trade-off arises between preserving detailed balance when using standard association kernels in the CR and VR models, versus allowing for a spatially uniform dissociation rate for the dissociation reaction. We show that preserving detailed balance with these association kernels requires a (decreased) spatially varying dissociation rate near domain boundaries. Such a mechanism can be realized in simulation methods by either using the spatially-varying dissociation rate near the domain boundary, or by using a constant rate, ignoring the domain boundary, and rejecting any unbinding events in which molecules are placed outside the domain. In Section~\ref{sect:kpfromkm} we reverse our approach, specifying an unbinding kernel for the Doi VR model that includes a spatially uniform dissociation rate, and demonstrating that the corresponding detailed balance preserving association kernel then involves a locally increased probability per time for reactants to react when near the boundary. Finally, in Section~\ref{sect:numerics} we demonstrate a simple numerical example illustrating differences that arise when using reaction product placement kernels that preserve detailed balance and do not preserve detailed balance in the vicinity of the domain boundary.

\section{PBSRD Models for Reversible Binding}
\label{sect:revBindNonStatTarget}
We consider the reversible $\textrm{A} + \textrm{B} \rightleftharpoons \textrm{C}$ reaction in a system with one $\textrm{A}$ molecule and one \textrm{B} molecule (or equivalently one \textrm{C} molecule).  Let $\vx$ denote the position of the $\textrm{A}$ molecule, $\vy$ the position of the $\textrm{B}$ molecule, and $\vz$ the position of the $\textrm{C}$ molecule. We assume the molecules diffuse within a bounded domain $\Omega \subset \R^d$ (with $d = 2$ or $d=3$). Let $p(\vx,\vy,t)$ denote the probability density the \textrm{A} and \textrm{B} molecules are unbound and located at $\vx$ and $\vy$ respectively at time $t$, and $\pb(\vz,t)$ the probability density the molecules are bound and the corresponding \textrm{C} molecule is located at $\vz$ at time $t$.

We denote by $\DA$, $\DB$ and $\DC$ the constant (positive) diffusivities of the \textrm{A}, \textrm{B} and \textrm{C} molecules respectively. With $\Id$ the $d$-dimensional identity matrix, we define two constant diffusivity matrices, given by the block matrices 
\begin{equation}
  \label{eq:diffusivityDefs}
  \begin{aligned}
    \Dmat &:= \begin{bmatrix}
      \DA  \Id & \vO \\
      \vO & \DB \Id    
    \end{bmatrix} \\
    \Dbmat &:= \DC \Id.
  \end{aligned}
\end{equation}
\begin{figure*}[tb]
  \includegraphics[width=.65\textwidth]{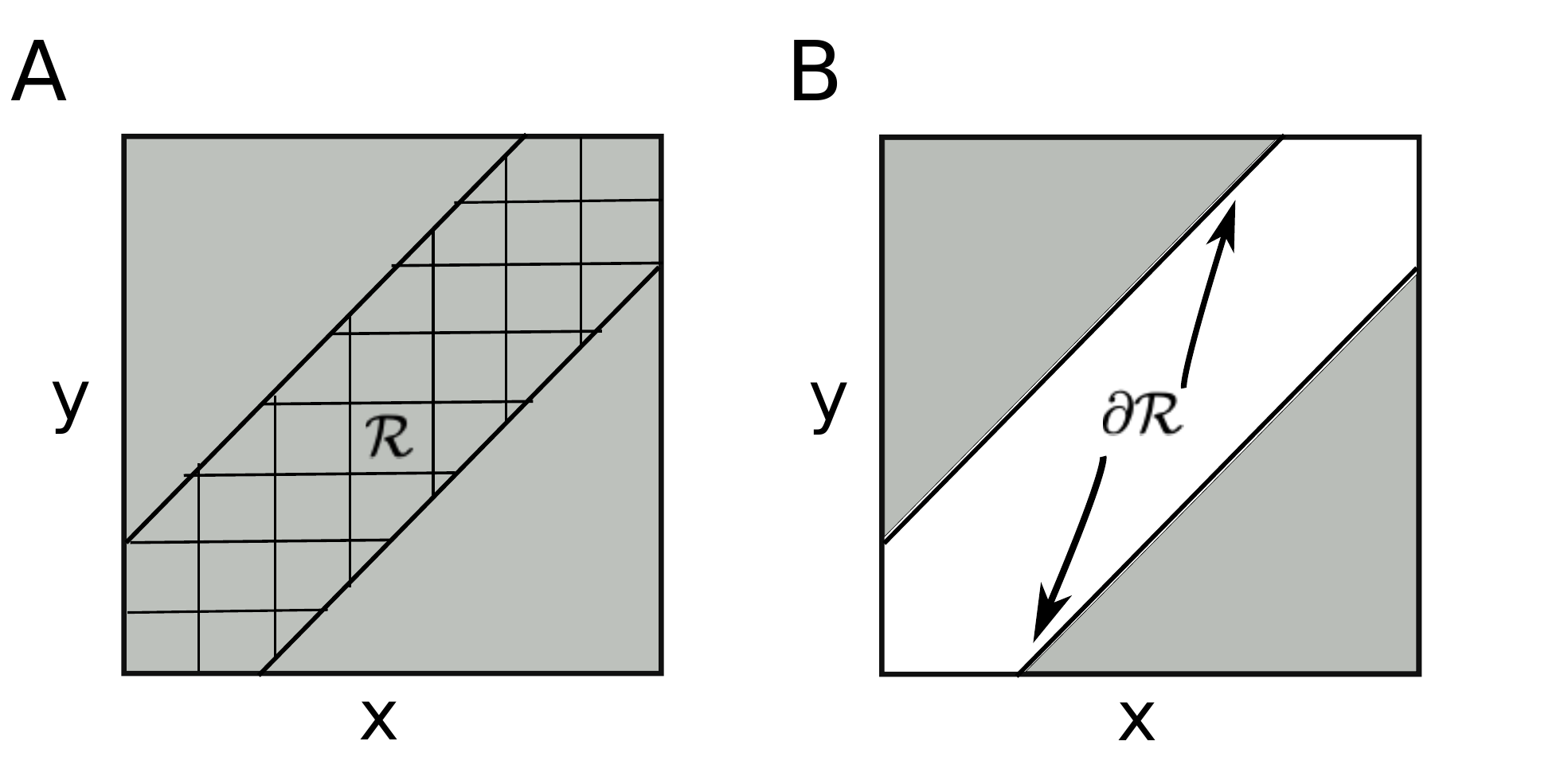}
  \caption{Reactive regions in the volume reactivity (VR) and contact reactivity (CR) PBSRD models when $\Omega$ is a one-dimensional interval. Here we plot the two-particle $(\vx,\vy)$ phase-space within the given interval. $\Omega^2$ corresponds to the entirety of the squares in both figures, while the shaded regions correspond to $\Omegaf^2$, the accessible region in which particles can diffuse. A) The VR model, where particles can react when within $\Rset$, the cross-hatched region. For some variants of the VR model this region can fill the entire square (e.g. Gaussian kernels~\eqref{eq:lambdaMultParticleGauss}), while for others (e.g. Doi kernels~\eqref{eq:lambdaMultParticleVR}) it may denote a subset of the square when $\vx$ and $\vy$ are sufficiently close. B) The CR model, where particles are excluded from the unshaded region, and can react or reflect apart when reaching the boundary (lines) $\dRset$ (e.g. the SCK model~\eqref{eq:lambdaMultParticleCR}).}
  \label{fig:Fig1}
\end{figure*}
With $\rb$ again labeling the reaction-radius, let $\Rset = \{ (\vx,\vy) \in \Omega^{2} \,\vert\, \abs{\vx-\vy} < \rb\}$, and denote by $\dRset = \{ (\vx,\vy) \in \Omega^{2} \,\vert\, \abs{\vx-\vy} = \rb\}$ the boundary of $\Rset$. In $\R^{2d}$, $\dRset$ corresponds to the set of $(\vx,\vy)$ pairs at which the association reaction can occur in the CR model. In the VR model the region in which a reaction can occur depends on the choice of rate functions. For example, the Doi model~\cite{DoiSecondQuantA,DoiSecondQuantB} only allows reactions for $(\vx,\vy)\in\Rset$, while a Gaussian interaction function model~\cite{Zhang2019} allows for reactions at all $(\vx,\vy) \in \Omega^2$. We denote the set of realizable reactant particle pair positions $(\vx,\vy) \in \Omega^2$ by
\begin{subequations}
  \begin{numcases}
    {\OSqD =}
    \Omega^2, &\text{VR model,}\label{eq:DoiDiffDomain}\\
    \Omega^2\setminus\Rset\cup\dRset, &\text{CR model.}\label{eq:SmolDiffDomain}
  \end{numcases}
\end{subequations}
The choice \eqref{eq:DoiDiffDomain} is used in the VR model, for which the Doi and $\lambda$-$\rho$ models are special cases~\cite{TeramotoDoiModel1967,DoiSecondQuantA,DoiSecondQuantB} where molecules react with a constant probability per time when positioned in $\mathcal{R}$. The choice~\eqref{eq:SmolDiffDomain} corresponds to the CR model, for which the Smoluchowski-Collins-Kimball (SCK) partial-absorption model is a special case~\cite{CollinsKimballPartialAdsorp,KeizerJPhysChem82}. In the SCK model two reactants either instantly react or are reflected apart upon reaching the boundary of the reactive region, $\dRset$. The preceding geometric regions are illustrated when $\Omega$ is an interval in Fig.~\ref{fig:Fig1}.

We will make use of an indicator function to denote the positions of realizable reactant pairs
\begin{equation} 
\label{eq:chiDef}
  \ind_{\OSqD}(\vx,\vy) = \begin{cases}
  1, &(\vx,\vy) \in \OSqD,\\
  0, &(\vx,\vy) \notin \OSqD.
  \end{cases}
\end{equation}
Finally, with $\nabla_{\vx,\vy} = (\nabla_{\vx},\nabla_{\vy})$ we then have the effective diffusion operators in $(\vx,\vy)$ and $\vz$:
\begin{equation}
  \label{eq:diffusionOps}
  \begin{aligned}
    \L &= \nabla_{\vx,\vy} \cdot \ind_{\OSqD}(\vx,\vy) \Dmat \nabla_{\vx,\vy}\\
    \Lb &= \nabla_{\vz} \cdot \Dbmat \nabla_{\vz}.
  \end{aligned}
\end{equation}

The forward association $\textrm{A} + \textrm{B} \to \textrm{C}$ reaction process is defined by a reaction kernel $\kp(\vz\vert\vx,\vy)$, corresponding to the probability density per unit time a reaction occurs creating a \textrm{C} molecule at $\vz$ given an \textrm{A} molecule at $\vx$ and a \textrm{B} molecule at $\vy$. We assume that $\kp(\vz\vert\vx,\vy)$ is specified through the factorization
\begin{align} \label{eq:kpfactorization}
  \kp(\vz\vert\vx,\vy) &= \mup(\vx,\vy) \kphat(\vz | \vx,\vy),
\end{align}
where $\mup(\vx,\vy)$ denotes the probability per unit time an \textrm{A} molecule at $\vx$ and a \textrm{B} molecule at $\vy$ attempt to react. $\kphat(\vz|\vx,\vy)$ denotes the probability density a reaction successfully occurs and creates a $\textrm{C}$ molecule at $\vz$, given that an $\textrm{A}$ molecule at $\vx$ and a $\textrm{B}$ molecule at $\vy$ attempted to react. Common choices for $\mup(\vx,\vy)$ are
\begin{align}
  &\begin{aligned} \label{eq:lambdaMultParticleVR}
  \mup(\vx,\vy) &= \lambda \ind_{\Rset}(\vx,\vy) \\
  & = \lambda \ind_{B_{\rb}(\vO)}(\vx-\vy),  
  \end{aligned}
  &\begin{aligned}
    \text{(Doi VR model)}
  \end{aligned}\\
  &\begin{aligned} \label{eq:lambdaMultParticleCR}
  \mup(\vx,\vy) &= \alpha \delta_{\dRset}(\vx,\vy) \\
  &= \sqrt{2} \alpha \delta_{\dB_{\rb}(\vO)}(\vx-\vy),  
  \end{aligned}
  &\begin{aligned}
    \text{(SCK CR Model)} 
  \end{aligned}
\end{align}
Here $B_{\rb}(\vO) = \{\abs{\vx} < \rb\}$ denotes the $d$-dimensional ball of radius $\rb$ about the origin, and $\dB_{\rb}(\vO) = \{\abs{\vx}=\rb\}$ the corresponding surface of the ball. In the Doi VR model, $\lambda$ corresponds to the probability per unit time the molecules may react when within $\rb$ of each other. In the SCK CR model, $\alpha$ corresponds to the absorption constant for the molecules to either react or reflect upon reaching a separation of $\rb$ (with units of length per time). The equivalence of the two $\delta$-surface measure representations given in~\eqref{eq:lambdaMultParticleCR} is shown in Appendix~\ref{app:proofOfEqCR}.

While~\eqref{eq:lambdaMultParticleVR}, which is discontinuous, is the most popular VR model, smooth interaction functions also arise in applications. For example, in modeling bimolecular reactions between membrane-bounded tethered signaling molecules with unstructured tails we derived and used the Gaussian interaction \cite{Goyette2017,Zhang2019}
\begin{equation}
  \mup(\vx,\vy) = \lambda\paren{\frac{3}{2\pi}}^{3/2}\frac{1}{\epsilon^3}e^{-\frac{3\abs{\vx-\vy}^2}{2\epsilon^2}}.
  \label{eq:lambdaMultParticleGauss}
\end{equation}

In both the VR and CR models, it is common to choose the placement kernel $\kphat(\vz|\vx,\vy)$ such that a newly created \textrm{C} molecule is placed on the line connecting the \textrm{A} and \textrm{B} molecules,
\begin{align} \label{eq:gammaStd}
  \kphat(\vz\vert\vx,\vy) &= \delta(\vz - \gamma \vx - (1-\gamma)\vy),
\end{align}
where $\gamma$ is a fixed value in  $\brac{0,1}$. One simple choice is $\gamma = \frac{1}{2}$, which corresponds to the midpoint between the two molecules. Another common choice, when $\DA$ and $\DB$ are constants, is to use the diffusion weighted center of mass~\cite{AndrewsBrayPhysBio2004},
\begin{equation}
  \gamma = \frac{\DB}{\DA + \DB}.
  \label{eq:plcmWeight}
\end{equation}
We note that with the choices \eqref{eq:gammaStd} and \eqref{eq:plcmWeight}, $\gamma = 0$ indicates that the $\textrm{B}$ molecule is not diffusing. Upon binding, the $\textrm{C}$ molecule is therefore placed at $\vy$. On the other hand, $\gamma = 1$ indicates that the $\textrm{A}$ molecule is not diffusing, and the $\textrm{C}$ molecule is then placed at $\vx$. Such choices would be appropriate if one of the $\textrm{A}$ or $\textrm{B}$ molecules represents a stationary or significantly more massive target. 

With the factorization~\eqref{eq:kpfactorization}, the probability that an attempted reaction between an \textrm{A} molecule at $\vx$ and a \textrm{B} molecule at $\vy$ successfully creates a \textrm{C} molecule (within the domain) is given by 
\begin{equation*}
  \int_{\Omega} \kphat(\vz\vert\vx,\vy) \, d\vz.
\end{equation*}
As such, the probability per time an $\textrm{A}$ molecule at $\vx$ and a $\textrm{B}$ molecule at $\vy$ successfully react to produce a $\textrm{C}$ molecule within the domain is
\begin{align*}  
  \kp(\vx,\vy) &:= \int_{\Omega} \kp(\vz\vert\vx,\vy) \, d\vz \\
  &\phantom{:}= \mup(\vx,\vy) \int_{\Omega} \kphat(\vz\vert\vx,\vy) \, d\vz.
\end{align*}
In freespace and periodic domains one usually has $  \int_{\Omega} \kphat(\vz\vert\vx,\vy) \, d\vz = 1$, i.e. the reaction always occurs successfully, so that $\kp(\vx,\vy) = \mup(\vx,\vy)$. As we discuss in Section~\ref{sect:application_DB}, this does not generally hold for standard association or dissociation reaction kernels in bounded domains with reflecting boundary conditions.

To completely specify the reaction-diffusion model, we must also give the unbinding kernel for the reverse dissociation $\textrm{C} \to \textrm{A} + \textrm{B}$ reaction. Let $\km(\vx,\vy\vert\vz)$ denote the probability density per time a reaction occurs producing an \textrm{A} molecule at $\vx$ and a \textrm{B} molecule at $\vy$ given a \textrm{C} molecule at $\vz$. Without loss of generality, assume $\gamma \neq 1$. As we will later show, detailed balance preserving choices for $\km(\vx,\vy|\vz)$ given the association reaction rate functions of the Doi~\eqref{eq:lambdaMultParticleVR} or SCK~\eqref{eq:lambdaMultParticleCR} models with placement density~\eqref{eq:gammaStd} are
\begin{align}  
 &\begin{aligned}\label{eq:mTwoPartDefVR}
   \km(\vx,\vy|\vz) = &\tfrac{\mum}{\abs{B_{(1-\gamma)\rb}(\vO)}} \ind_{B_{(1-\gamma)\epsilon}(\vO)}(\vx-\vz) \delta\paren{\vy - \frac{\vz-\gamma\vx}{1-\gamma}}, 
 \end{aligned}
 &\begin{aligned}
   \text{(Doi VR)}     
 \end{aligned}\\
 &\begin{aligned} \label{eq:mTwoPartDefCR}
   \km(\vx,\vy|\vz) = &\tfrac{\mum}{\abs{\partial B_{(1-\gamma)\rb}(\vO)}}\delta_{\partial B_{(1-\gamma)\epsilon}(\vO)}(\vx-\vz)
   \delta\paren{\vy - \frac{\vz-\gamma\vx}{1-\gamma}},
 \end{aligned}
 &\begin{aligned}
   \text{(SCK CR)}
 \end{aligned}
\end{align}
assuming $\mum$ is chosen appropriately. Here $B_{(1-\gamma)\rb}(\vO)$ denotes the ball of radius $(1-\gamma)\rb$ centered at the origin, $\partial B_{(1-\gamma)\rb}(\vO)$ the sphere of radius $(1-\gamma)\rb$ centered at the origin (i.e. boundary of the ball), and $\abs{B_{(1-\gamma)\rb}(\vO)}$ and $ \abs{\partial B_{(1-\gamma)\rb}(\vO)}$ their respective sizes (e.g. volume and surface area respectively when in three-dimensions). $\mum$ represents the (constant) probability per time a $\textrm{C}$ molecule attempts to dissociate into \textrm{A} and \textrm{B} molecules. As written, the Doi unbinding kernel corresponds to sampling the position of the \textrm{A} molecule within a ball of radius $(1-\gamma)\rb$ about the \textrm{C} molecule, and then placing the \textrm{B} molecule  by reflection on the line connecting the \textrm{A} and \textrm{C} molecules. The SCK CR model modifies this process by sampling the position of the \textrm{A} molecule on the surface of the sphere. Note, for the Doi VR (SCK CR) model one could equivalently sample the position of the \textrm{B} molecule within (on) a ball of radius $\gamma\rb$, and then place the \textrm{A} molecule on the line connecting the \textrm{B} and \textrm{C} molecules. 

Finally, with $\km(\vx,\vy|\vz)$ specified we define $\km(\vz)$ to be the probability per time that a \textrm{C} molecule successfully dissociates at $\vz$, producing \textrm{A} and \textrm{B} molecules within $\Omega$. We have
\begin{equation*}
  \km(\vz) = \int_{\Omega^2} \km(\vx,\vy\vert\vz) \, d\vx \, d\vy.
\end{equation*}
In freespace and periodic domains one usually has $\km(\vz) = \mum$, i.e. the reaction always occurs successfully with fixed rate constant, $\mum$~\cite{FrohnerNoe2018,Donevetal2018}. 

Given the preceding definitions, our general model for the two-particle reversible $\textrm{A} + \textrm{B} \rightleftharpoons \textrm{C}$ reaction with $(\vx,\vy) \in \Omega^{2}$ and $\vz \in \Omega$ is
\begin{subequations} \label{eq:bothDiffuseEqs} 
 \begin{align}
   \begin{aligned} \label{eq:bothDiffuseEqsrho}
   \ind_{\OSqD}(\vx,\vy) \PD{p}{t}(\vx,\vy,t) &= \L p(\vx,\vy,t) - \kp(\vx,\vy) p(\vx,\vy,t) 
   + \int_{\Omega} \km(\vx,\vy\vert\vz) \pb (\vz,t) \, d \vz, 
   \end{aligned}\\
   \begin{aligned}\label{eq:bothDiffuseEqsrhob}
   \PD{\pb}{t}(\vz,t) &= \Lb \pb(\vz,t) - \km(\vz) \pb(\vz,t) 
   + \int_{\Omega^{2}} \kp(\vz\vert\vx,\vy) p(\vx,\vy,t) \, d\vx \, d \vy. 
 \end{aligned}
 \end{align}
\end{subequations}
Here we assume a reflecting zero Neumann boundary condition on $\partial \Omega$ in each coordinate respectively (i.e. $\vx$, $\vy$ and $\vz$), 
\begin{align*}
  \nabla_{\vx,\vy} p(\vx,\vy,t) \cdot \veta(\vx,\vy) &= 0, &&(\vx,\vy) \in \partial(\Omegaf^2),\\
   \nabla_{\vz} \pb(\vz,t) \cdot \veta_{\textrm{b}}(\vz) &= 0, &&\vz \in \partial \Omega,
\end{align*}
where $\partial(\Omegaf^2)$ denotes the two-particle phase-space boundary, $\veta(\vx,\vy)$ denotes the unit outward normal to this boundary at $(\vx,\vy)$, and $\veta_{\textrm{b}}(\vz)$ denotes the unit outward normal to $\partial\Omega$ at $\vz$. For the VR model $\partial(\Omegaf^2) = \partial \Omega \times \partial \Omega$, i.e. the phase-space boundary corresponding to each particle reflecting off the domain boundary $\partial \Omega$. In Fig.~\ref{fig:Fig1} this corresponds to the entire square bounding the domain. In the CR model $\partial (\Omegaf^2)$ corresponds to the portion of $\partial \Omega \times \partial \Omega$ that is outside $\mathcal{R} \cup \partial \mathcal{R}$. In Fig.~\ref{fig:Fig1} this corresponds to the portion of the square bounding the domain that borders the shaded region that represents $\Omegaf^2$. Finally, we assume the initial conditions
\begin{align*}
  p(\vx,\vy,0) &= p_0(\vx,\vy), & \pb(\vz,0) = \pbO(\vz),
\end{align*}
where $p_0(\vx,\vy)$ is zero outside $\OSqD$. We also assume that $p_0$ and $\pbO$ define a proper probability  distribution so that 
\begin{equation*}
  \int_{\OSqD} p_0(\vx,\vy) \, d\vx \, d\vy + \int_{\Omega} \pbO(\vz) \, d\vz = 1.
\end{equation*}
Integrating~\eqref{eq:bothDiffuseEqs} over $(\vx,\vy) \in \Omega^{2}$ and $\vz \in \Omega$, and using the definitions of $\km(\vz)$ and $\kp(\vx,\vy)$, this normalization of the initial conditions immediately implies that probability is conserved for all times:
\begin{equation*}
  \int_{\OSqD} p(\vx,\vy,t) \, d\vx \, d \vy
  + \int_{\Omega} \pb(\vz,t) \, d\vz = 1.
\end{equation*}

We note that \eqref{eq:bothDiffuseEqs} encompasses both the general VR and CR models. In particular, in Appendix~\ref{sect:WeakFormContactRx} we show the weak form~\cite{EvansPDEs2ed,SchussStochProcBook2010} of \eqref{eq:bothDiffuseEqs} with SCK rate kernels \eqref{eq:lambdaMultParticleCR}, \eqref{eq:gammaStd} and \eqref{eq:mTwoPartDefCR} is equivalent to the weak form of the standard representation for the SCK model (in which the association reaction is represented by a partial-absorption boundary condition~\cite{AgmonSzaboRevRx1990,CollinsKimballPartialAdsorp,KeizerJPhysChem82}).

\section{Steady-State Detailed Balance and Equilibrium Solutions}
\label{sect:SSDBandEqSoln}
At steady-state we find the solutions to~\eqref{eq:bothDiffuseEqs}, $\peq(\vx,\vy)$ and $\pbeq(\vz)$, satisfy
\begin{equation} \label{eq:bothEqsSS}
  \begin{aligned}
    0 &= \L \peq(\vx,\vy) - \kp(\vx,\vy) \peq(\vx,\vy) + \int_{\Omega} \km(\vx,\vy\vert\vz) \pbeq (\vz) \, d \vz, \\
    0 &= \Lb \pbeq(\vz) - \km(\vz) \pbeq(\vz) + \int_{\Omega^{2}} \kp(\vz\vert\vx,\vy) \peq(\vx,\vy) \, d\vx \, d \vy,
  \end{aligned}
\end{equation}
with a reflecting zero Neumann boundary condition in each coordinate on $\partial \Omega$ and the normalization
\begin{equation*}
\int_{\OSqD} \peq(\vx,\vy) \, d\vx \, d \vy + \int_{\Omega} \pbeq(\vz) \, d\vz = 1.
\end{equation*}

As discussed in the introduction, we expect that the steady-state for the reversible $\textrm{A} + \textrm{B} \leftrightarrows \textrm{C}$ reaction is a state of thermodynamic equilibrium, with (pointwise) detailed balance of spatial reaction fluxes holding for the reactive terms, i.e.
\begin{equation}
  \label{eq:detailedBalanceRxEq}
  \kp(\vz\vert\vx,\vy) \peq(\vx,\vy) = \km(\vx,\vy\vert\vz) \pbeq(\vz).
\end{equation}
By integrating~\eqref{eq:detailedBalanceRxEq} in $\vz$
(resp. $(\vx,\vy)$), we find that the reactive terms
in~\eqref{eq:bothEqsSS} cancel out,
\begin{align*}
  \kp(\vx,\vy) \peq(\vx,\vy) &= \int_{\Omega} \km(\vx,\vy\vert\vz) \pbeq (\vz) \, d \vz,\\
  \km(\vz) \pbeq(\vz) &= \int_{\Omega^{2}} \kp(\vz\vert\vx,\vy) \peq(\vx,\vy) \, d\vx \, d \vy.
\end{align*}
These then imply that $\L \peq = 0$ on $\OSqD$ and
$\Lb \pbeq = 0$ on $\Omega$ which, together with the assumed reflecting zero Neumann boundary conditions on $\partial \Omega^2$ and $\partial \Omega$, gives that $\peq$ and $\pbeq$ are
constant. Using~\eqref{eq:detailedBalanceRxEq}, we see that for the system to be consistent with the principle of detailed balance, we must choose $\kp(\vz\vert\vx,\vy)$ and $\km(\vx,\vy\vert\vz)$ such that
\begin{equation} \label{eq:detailedBalancecoefRelation}
  \kp(\vz\vert\vx,\vy) \propto \km(\vx,\vy\vert\vz).
\end{equation}
We may define the dissociation constant of the reaction, $\Kd$, to be the constant of proportionality, so that \begin{equation}
    \km(\vx,\vy\vert\vz) = \Kd\kp(\vz\vert\vx,\vy).
    \label{eq:detailedBalancekpkm}
\end{equation}
Note, for detailed balance to hold this implies that $\Kd$ is also given by
\begin{equation}
    \label{eq:detailedBalancekpkmints}
    \Kd = \frac{\iint_{\OSqD} \int_{\Omega} \km(\vx,\vy\vert\vz) \, d\vz \, d\vy \, d\vx}{\iint_{\OSqD}\int_{\Omega} \kp(\vz\vert\vx,\vy) \, d\vz \, d\vy \, d\vx}.
\end{equation}

Together with the normalization condition that
\begin{equation*}
  \peq \abs{\OSqD} + \pbeq \abs{\Omega} = 1,
\end{equation*}
we obtain
\begin{theorem}
\label{thm:SSKd}
  When the detailed balance statement~\eqref{eq:detailedBalanceRxEq}
  is satisfied, on their appropriate domains of definition, $\OSqD$
  and $\Omega$ respectively,
  \begin{align} \label{eq:twoPartEqSCKSoluts}
    \peq &= \frac{\Kd}{\abs{\Omega} + \Kd \abs{\OSqD}}, 
    & \pbeq &= \frac{1}{\abs{\Omega} + \Kd \abs{\OSqD}}.
  \end{align}
  In the VR model $\OSqD = \Omega^2$ so that
  this simplifies to
  \begin{align} \label{eq:twoPartEqVRSoluts}
    \peq &= \frac{\Kd}{\abs{\Omega}(1 + \Kd \abs{\Omega})}, 
    & \pbeq &= \frac{1}{\abs{\Omega}(1 + \Kd \abs{\Omega})}.
  \end{align}
\end{theorem}
Let $\bar{P} = \int_{\Omega^2}\peq(\vx,\vy)\, d\vx \, d\vy$ and $\bar{P}_{\textrm{b}} = \int_{\Omega}\pbeq(\vz)\, d\vz$ denote the probabilities to be in the unbound and bound state respectively. For the VR model~\eqref{eq:twoPartEqVRSoluts} gives
\begin{align} \label{eq:ssprobs}
\bar{P} &= \frac{\Kd\abs{\Omega}}{1 + \Kd \abs{\Omega}}, & \bar{P}_{\textrm{b}} &= \frac{1}{1 + \Kd \abs{\Omega}},
\end{align}
which are identical to what one obtains in a well-mixed stochastic chemical kinetics model, see Appendix~\ref{sect:SSProbKd}.

\section{Detailed Balance in Bounded Domains}
\label{sect:application_DB}
Unless $\Omega$ is convex, the line segment connecting the positions of an \textrm{A} molecule and a \textrm{B} molecule may leave the domain. As such, when using the placement density given by~\eqref{eq:gammaStd} the position of a new \textrm{C} molecule may fall outside $\Omega$ if $\gamma \in (0,1)$. One approach to address this issue is that used in the Smoldyn simulator~\cite{AndrewsBrayPhysBio2004}, where a straight line is drawn from one molecule to the other. If the straight line crosses the domain boundary, then the \textrm{A} and \textrm{B} molecules are not allowed to react. Let $\ell_{\vx,\vy} = \{s \vx + (1-s)\vy \mid s \in [0,1]\}$ be the straight line connecting an \textrm{A} molecule at $\vx$ and a \textrm{B} molecule at $\vy$. $\kphat(\vz|\vx,\vy)$ is modified to
\begin{equation} \label{eq:SmoldynCPlacement}
  \kphat(\vz|\vx,\vy) = 
  \begin{cases}
    \delta(\vz - (\gamma \vx + (1-\gamma) \vy)), &\text{if } \ell_{\vx,\vy} \subseteq \Omega \\
    0, &\text{otherwise.}
  \end{cases}
\end{equation}
With this choice
\begin{equation*}
  \begin{aligned}  
  \prob \brac{\text{reaction is accepted}} &=  
  \int_{\Omega} \kphat(\vz\vert\vx,\vy) \, d\vz  \\
  &= \begin{cases}
    1, &\text{if } \ell_{\vx,\vy} \subseteq \Omega \\
    0, &\text{otherwise,}
  \end{cases}
  \end{aligned}
\end{equation*}
so that only binding events for which the line segment is within $\Omega$ are accepted.  In contrast, with the choice~\eqref{eq:gammaStd},
\begin{equation*}
\int_{\Omega} \kphat(\vz\vert\vx,\vy) \, d\vz = \ind_{\Omega}(\gamma \vx + (1-\gamma) \vy),
\end{equation*}
so that binding events where $\vz$ would be placed outside $\Omega$ are rejected (i.e. are not allowed to proceed). 
Note, if $\Omega$ is convex then $\gamma \vx + (1-\gamma)\vy \in \Omega$ for all $\vx$ and $\vy$ in $\Omega$. Association reactions are therefore always successful and the choices~\eqref{eq:gammaStd} and~\eqref{eq:SmoldynCPlacement} are identical.

\begin{figure*}
  \includegraphics[width=.75\textwidth]{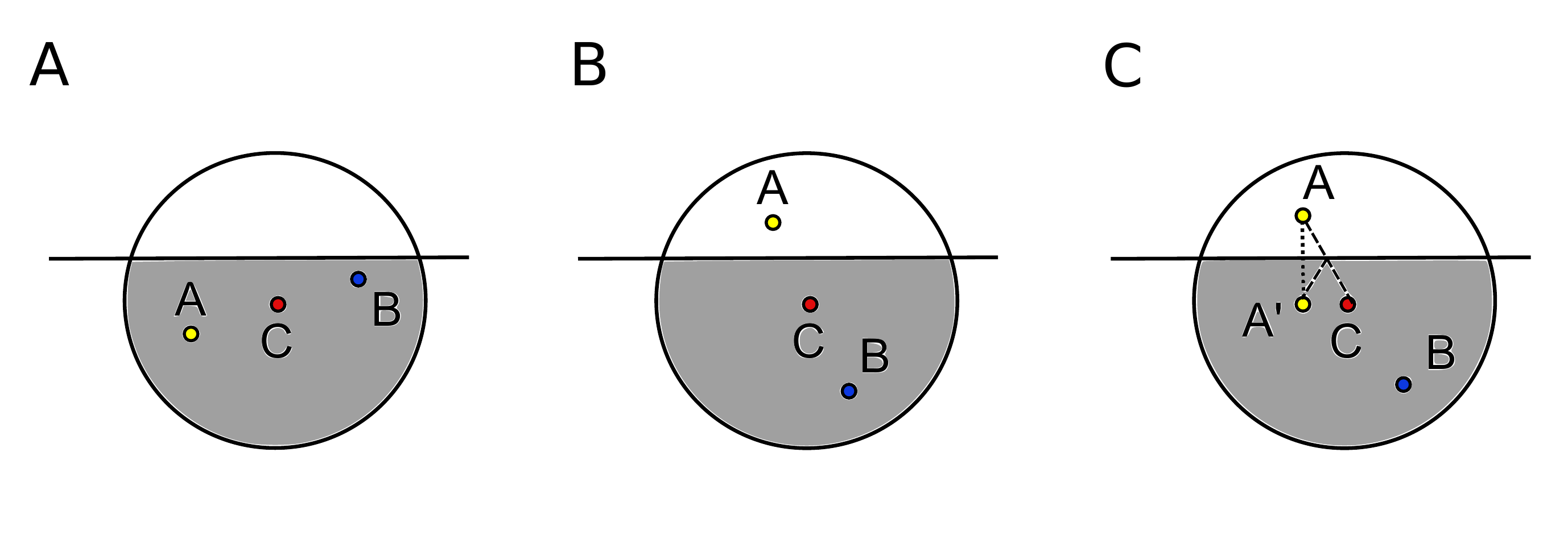}
  \caption{Rejection vs. reflection placement mechanisms for the $\textrm{C} \to \textrm{A} + \textrm{B}$ reaction in the standard Doi VR model considered in Section~\ref{sect:application_DB}. A) and B) illustrate the acceptance/rejection mechanism which preserves detailed balance. Here positions for the $\textrm{A}$ and $\textrm{B}$ molecules are chosen ignoring the boundary. The approach in Section~\ref{sect:VR_DB} and~\ref{sect:dbreject} samples the position of the $\textrm{A}$ particle, $\vx$, uniformly within the ball of radius $(1-\gamma)\rb$ about $\vz$, the position of the $\textrm{C}$ molecule. The position of the $\textrm{B}$ molecule, $\vy$, is then chosen by reflection as $\vy = (\vz - \gamma \vx) / (1-\gamma)$. The circle corresponds to the ball of radius $(1-\gamma)\varepsilon$ in which the $\textrm{A}$ molecule can be placed, with the gray region representing the portion of the ball within the domain, $\Omegaf$. If both molecules end up within the domain, $\Omegaf$, their positions are accepted (case A), otherwise the unbinding event is rejected (case B). C) illustrates a mirror-like reflection process that can be used if either the product $\textrm{A}$ or $\textrm{B}$ molecules end up outside the domain~\cite{AndrewsBrayPhysBio2004}. The particle is initially placed ignoring the boundary, but then reflected off the surface relative to the normal to the surface at the point of intersection of the line connecting the position of the $\textrm{C}$ molecule and the product molecule. In C) the \textrm{A} molecule is reflected to $\textrm{A}'$. Figures are draw with $\gamma = \tfrac{1}{2}$.}
  \label{fig:Fig2}
\end{figure*}
Analogously, for the backward dissociation reaction when $\vz$ is near $\partial \Omega$ there is the possibility that one or both of the sampled product particle positions, $\vx$ and $\vy$, end up outside the domain. One approach that has been used in existing simulators is that developed by Smoldyn~\cite{AndrewsBrayPhysBio2004}. There, if either of the \textrm{A} or \textrm{B} molecules are placed outside the domain, they are then reflected across the boundary surface back into $\Omega$, so that unbinding reactions are always successful. However, even if $\Omega$ is assumed convex so that association reactions are always accepted, we show in Section~\ref{sect:DBrejection} that a detailed balance preserving Doi or SCK unbinding kernel with the forward placement kernel~\eqref{eq:gammaStd} results in a reduced unbinding rate, which is equivalent to rejection of some unbinding events near $\partial \Omega$. As such, a Doi or SCK unbinding kernel with reflection violates detailed balance near domain boundaries. The rejection-based kernel we discuss preserves detailed balance, but in contrast to a reflection model in which unbinding is always successful, requires the use of a reduced and spatially-varying unbinding rate near domain-boundaries. For the Doi model of the next subsection the rejection versus reflection unbinding processes are illustrated in Fig.~\ref{fig:Fig2}.

Whether a decreased unbinding rate is physically appropriate will depend on the underlying physics for the unbinding reaction. For example, if the domain boundary can be thought of as impeding two particles from separating apart, using a reduced reaction-rate may be physically realistic. In contrast, if being near the domain boundary should have no effect on the timescale for molecules to dissociate, a constant unbinding rate would be more appropriate. In Section~\ref{sect:kpfromkm} we take the opposite perspective to~\eqref{eq:detailedBalancekpkm}, using this formula to choose $\kp$ given a choice for $\km$ in a Doi-type model with a spatially uniform unbinding rate, i.e. $\km(\vz) = \mum$. In this special case we show that the detailed balance relation then requires $\kp$ to give an increased probability per time for a reaction to occur when reactants are near the domain boundary. As such, preserving detailed balance, spatially uniform unbinding rates, and standard association kernels, $\kp(\vx,\vy)$, appears to require alternative product placement kernels from those commonly used in applications.

We begin in the next section by illustrating several commonly used choices of rate functions and placement kernels, demonstrating for which pairs detailed balance holds or does not hold. In Section~\ref{sect:DBrejection} we then show that the standard Doi VR and SCK CR model choices for the forward reaction  kernel require rejection of some unbinding events near the domain boundary to satisfy detailed balance. In Section~\ref{sect:kpfromkm} we show that assuming detailed balance and a spatially uniform unbinding rate for a version of the Doi model results in an increased probability per time for the forward reaction to occur in the vicinity of the domain boundary. Finally, in Section~\ref{sect:numerics} we illustrate via numerical simulations how differences can arise when using the detailed balance preserving rejection-based kernel versus a reflection-based approach.


\subsection{Detailed Balance Determines the Unbinding Kernel} \label{sect:DB_for_common_models}
The detailed balance condition given by \eqref{eq:detailedBalancekpkm} implicitly defines one of the binding or unbinding kernels in terms of the other through the dissociation constant, $\Kd$. As such, any detailed balance preserving binding and unbinding kernels must have the same functional form. We now show that for the standard choices of $\kp(\vz\vert\vx,\vy)$, several standard choices for $\km(\vx,\vy\vert\vz)$ will preserve detailed balance for all $(\vx,\vy) \in \OSqD$ and all $\vz \in \Omega$ when including \emph{rejection} near $\partial\Omega$. We further demonstrate that several other choices that have been used in the literature will violate detailed balance for particles near the domain boundary.
 
 \subsubsection{The Standard Doi VR Model}
\label{sect:VR_DB}
In the Doi VR model with a zero reflecting Neumann boundary condition, a common choice of reactive terms are given by~\eqref{eq:lambdaMultParticleVR}, \eqref{eq:gammaStd}, and~\eqref{eq:mTwoPartDefVR},
    \begin{align}
        \kp(\vz \vert \vx, \vy) &= \lambda\ind_{B_{\rb}(\vO)}(\vx - \vy)\delta(\vz - \gamma\vx - (1-\gamma)\vy), \label{eq:fwdRxDoi}\\
        \km(\vx,\vy \vert \vz) &= \mum\delta\paren{\vy - \frac{\vz-\gamma\vx}{1-\gamma}}\frac{\ind_{B_{(1-\gamma)\rb}(\vO)}(\vx-\vz)}{\abs{B_{(1-\gamma)\rb}(\vO)}} \label{eq:bwdRxDoi}.
    \end{align}
In~\eqref{eq:fwdRxDoi} the $\delta$ function implies that $\vy = (\vz - \gamma \vx)/(1 - \gamma)$ so that
\begin{equation} \label{eq:fwdtobackrates}  
 \begin{aligned}
   \ind_{B_{\rb}(\vO)}(\vx - \vy) \delta(\vz - \gamma\vx - (1-\gamma)\vy) 
   &= (1-\gamma)^{-d} \ind_{B_{\rb}(\vO)}\paren{\frac{\vx-\vz}{1-\gamma}} \delta\paren{\vy - \frac {\vz-\gamma\vx}{1-\gamma}} \\
   &= (1-\gamma)^{-d} \ind_{B_{(1-\gamma)\rb}(\vO)}\paren{\vx-\vz} \delta\paren{\vy - \frac{\vz-\gamma\vx}{1-\gamma}}.
 \end{aligned}
\end{equation}
Here we have used that $\delta(a\vx) = \delta(\vx)/\abs{a}^d$. The two rate kernels are therefore proportional, and detailed balance will hold if $\lambda$ and $\mum$ are chosen appropriately. 

The constant of proportionality, $\Kd$, can be determined using \eqref{eq:detailedBalancekpkmints}, allowing us to write $\km$ in terms of $\kp$ for all $(\vx,\vy) \in \Omega^2$ and all $\vz \in \Omega$. Assuming the rate constants are non-zero, evaluating~\eqref{eq:detailedBalancekpkmints} we find
\begin{equation}
  \begin{aligned}
    \Kd 
    &= \frac{\int_{\Omega^3}  \mum\delta\paren{\vy - \frac{\vz-\gamma\vx}{1-\gamma}}\frac{\ind_{B_{(1-\gamma)\rb}(\vO)}(\vx-\vz)}{\abs{B_{(1-\gamma)\rb}(\vO)}} \, d\vz \, d\vx \,   d\vy}
    {\int_{\Omega^3} \lambda\ind_{B_{\rb}(\vO)}(\vx - \vy)\delta(\vz - \gamma\vx - (1-\gamma)\vy) \, d\vz \, d\vx \, d\vy}, \\
    &= \frac{\frac{\mum}{\abs{B_{\rb}(\vO)}}\int_{\Omega^2}\ind_{B_{\rb}(\vO)}(\vx-\vy)\ind_{\Omega}(\gamma\vx + (1-\gamma)\vy) \,d\vx \, d\vy}
    {\lambda\int_{\Omega^2}\ind_{B_{\rb}(\vO)}(\vx - \vy)\ind_{\Omega}(\gamma\vx + (1-\gamma)\vy) \, d\vx \, d\vy}  \\
    &= \frac{\mum}{\lambda\abs{B_{\rb}(\vO)}}.
  \end{aligned}
  \label{eq:VR_Kd}
\end{equation}
Here we have again used that $\delta(a\vx) = \delta(\vx)/\abs{a}^d$, and used that $\abs{B_{(1-\gamma)\rb}(\vO)} = (1-\gamma)^d\abs{B_{\rb}(\vO)}$. From~\eqref{eq:fwdRxDoi} and~\eqref{eq:fwdtobackrates} we may then write $\km$ in terms of $\kp$ as
\begin{equation*}
  \begin{aligned}
    \Kd\kp(\vz \vert \vx,\vy) &= \mum\frac{\ind_{B_{\rb}(\vO)}(\vx - \vy)}{\abs{B_{\rb}(\vO)}}\delta(\vz - \gamma\vx - (1-\gamma)\vy), \\
    &= \mum\delta\paren{\vy - \frac{\vz-\gamma\vx}{1-\gamma}}\frac{\ind_{B_{(1-\gamma)\rb}(\vO)}(\vx-\vz)}{\abs{B_{(1-\gamma)\rb}(\vO)}} \\
    &= \km(\vx,\vy \vert \vz).
  \end{aligned}      
\end{equation*}
In Section~\ref{sect:dbreject}, considering the association kernel as fixed we show that the unbinding kernel~\eqref{eq:bwdRxDoi} implies the unbinding rate $\km(\vz)$ is smaller than $\mum$ for $\vz$ near $\partial \Omega$, while equal to $\mum$ away from the boundary. We then illustrate how this can be realized in simulations by allowing molecules to unbind at rate $\mum$ while ignoring the domain boundary, and then rejecting unbinding events that lead to a molecule ending up outside the domain. See also Fig.~\ref{fig:Fig2}A and Fig.~\ref{fig:Fig2}B for an illustration of the acceptance/rejection process.

\subsubsection{An alternative VR Model}
\label{sect:AlternativeVR}
We consider an alternative VR model for the reversible $\textrm{A} + \textrm{B} \rightleftharpoons \textrm{C}$ reaction, based on the model proposed in \cite{Donevetal2018}. As in the Doi VR model the \textrm{A} molecule reacts with the \textrm{B} molecule with rate $\lambda$ when their separation is within a distance $\rb$. Let $\rho$ be the probability of placing the \textrm{C} molecule at the position of the \textrm{A} molecule upon binding, and let $(1-\rho)$ be the probability of placing the \textrm{C} molecule at the position of the \textrm{B} molecule. Hence, upon binding ($\textrm{A} + \textrm{B} \to \textrm{C}$), one of the two molecules is chosen at random with probability $\rho$ (or $1-\rho$) and turns into the \textrm{C} molecule while the other molecule disappears. One simple choice of the selection probability $\rho$ is to take $\rho = \frac{1}{2}$ as in \cite{Donevetal2018}. For a general $\rho$, $\kp(\vz\vert\vx,\vy)$ is 
\begin{equation}\label{eq:fwdRxAltVR} 
 \kp(\vz\vert\vx,\vy) 
 = \lambda\ind_{B_{\rb}(\vO)}(\vx - \vy)\brac{\rho \,\delta(\vz - \vx) + (1-\rho)\delta(\vz-\vy)}.
\end{equation}
Unbinding ($\textrm{C} \to \textrm{A} + \textrm{B}$) is assumed to occur with a rate $\mum$. Upon unbinding a product \textrm{A}(or \textrm{B}) is chosen at random with probability $\rho$ (or $1-\rho$) and is placed at the position of \textrm{C}. The other product is placed uniformly in a sphere centered at the position of \textrm{C} with radius $\rb$. $\km(\vx,\vy\vert\vz)$ becomes
\begin{equation}  
 \begin{aligned}\label{eq:bwdRxAltVR}
 \km(\vx, \vy \vert \vz) = \mum\Bigg[&\rho\frac{\ind_{B_{\rb}(\vO)}(\vy - \vz)}{\abs{B_{\rb}(\vO)}}\delta(\vx - \vz)
  + (1-\rho)\frac{\ind_{B_{\rb}(\vO)}(\vx - \vz)}{\abs{B_{\rb}(\vO)}}\delta(\vy - \vz)\Bigg].
 \end{aligned}
\end{equation}

These reaction rate functions are shown to be detailed-balance preserving with periodic boundary conditions in \cite{Donevetal2018}. We now show that the detailed balance condition~\eqref{eq:detailedBalancekpkm} holds for all $(\vx,\vy) \in \Omega^2$ and all $\vz \in \Omega$ with reflecting boundary conditions. With \eqref{eq:fwdRxAltVR} and \eqref{eq:bwdRxAltVR}, and using the same scaling properties as in Section~\ref{sect:VR_DB}, when detailed balance holds the dissociation constant $\Kd$ is given by
  \begin{equation*}
  \begin{aligned}
    \Kd &= \frac{\int_{\Omega^2} \int_{\Omega} \mum\brac{\rho\frac{\ind_{B_{\rb}(\vO)}(\vy - \vz)}{\abs{B_{\rb}(\vO)}}\delta(\vx - \vz) + (1-\rho)\frac{\ind_{B_{\rb}(\vO)}(\vx - \vz)}{\abs{B_{\rb}(\vO)}}\delta(\vy - \vz)}\, d\vz \, d\vx \, d\vy}
      {\int_{\Omega^2} \int_{\Omega} \lambda\ind_{B_{\rb}(\vO)}(\vx-\vy)\brac{\rho\delta(\vz - \vx) + (1-\rho)\delta(\vz - \vy)} \, d\vz \, d\vx \, d\vy}, \\
    &= \frac{\mum}{\lambda\abs{B_{\rb}(\vO)}}\frac{\int_{\Omega^2}\brac{\rho\ind_{B_{\rb}(\vO)}(\vy - \vx) + (1-\rho)\ind_{B_{\rb}(\vO)}(\vx - \vy)}\, d\vx \, d\vy}
      {\int_{\Omega^2}\ind_{B_{\rb}(\vO)}(\vx-\vy) \, d\vx \, d\vy} \\
      &= \frac{\mum}{\lambda\abs{B_{\rb}(\vO)}}.
    \end{aligned}
  \end{equation*}
This immediately gives
\begin{equation*}
    \begin{aligned}
      \Kd\kp(\vz \vert \vx,\vy) 
      &= \mum\brac{\rho\frac{\ind_{B_{\rb}(\vO)}(\vx-\vy)}{\abs{B_{\rb}(\vO)}}\delta(\vz-\vx) + (1-\rho)\frac{\ind_{B_{\rb}(\vO)}(\vx-\vy)}{\abs{B_{\rb}(\vO)}}\delta(\vz-\vy)}, \\
      &= \mum\brac{\rho\frac{\ind_{B_{\rb}(\vO)}(\vy - \vz)}{\abs{B_{\rb}(\vO)}}\delta(\vx - \vz) + (1-\rho)\frac{\ind_{B_{\rb}(\vO)}(\vx - \vz)}{\abs{B_{\rb}(\vO)}}\delta(\vy - \vz)} \\
      &= \km(\vx,\vy \vert \vz),
    \end{aligned}
\end{equation*}
showing detailed balance.

In Appendix~\ref{sect:DBrejectionAltVR}, considering the association kernel as fixed we show that detailed balance holding corresponds to a decreased effective unbinding rate, $\km(\vz)$, for $\vz$ near $\partial \Omega$. This can once again be realized in simulations by using a rejection kernel near $\partial \Omega$, initially ignoring the domain boundary when a \textrm{C} molecule dissociates into \textrm{A} and \textrm{B} molecules, but rejecting unbinding events that lead to a molecule outside the domain. 

\subsubsection{The SCK CR Model}
\label{sect:CR_DB}
For the SCK CR model, one can repeat the same calculation as for the Doi VR model with the reaction terms \eqref{eq:lambdaMultParticleVR} and \eqref{eq:mTwoPartDefVR} replaced by \eqref{eq:lambdaMultParticleCR} and \eqref{eq:mTwoPartDefCR}. We show in Appendix~\ref{app:CRKd} that the dissociation constant, $\Kd$, for the SCK CR model is
\begin{equation*}
  \Kd = \frac{\mum}{\alpha\sqrt{2}\abs{\partial B_{\rb}(\vO)}},
\end{equation*} 
and detailed balance is preserved. As in the previous two sections, considering the association kernel as fixed this again corresponds to having a non-uniform unbinding rate, $\km(\vz) < \mum$ for $\vz$ near $\partial \Omega$, as we prove in Appendix~\ref{sect:DBrejectionSCK}. In simulations, this kernel can again be implemented by ignoring the boundaries when placing molecules, but rejecting those reaction events that lead to molecules outside the domain. 

\subsubsection{Reflection across the Boundary}
\label{sect:SmoldynDB}
Here we consider an alternative reaction kernel to overcome the problem of having to account for the boundary when a \textrm{C} molecule dissociates into \textrm{A} and \textrm{B} molecules. The boundary is initially ignored when placing the products, with those placed outside the domain then reflected across the boundary back into $\Omega$, see Fig.~\ref{fig:Fig2}C. The kernel is inspired by how Smoldyn handles products that are placed outside a domain, and represents a mirror-like reflection process~\cite{AndrewsBrayPhysBio2004}. With this reflection kernel unbinding events are always successful and we have
  \begin{equation}
  \label{eq:SmoldynUnbindingProb}
    \km(\vz) = \int_{\Omega^2} \km(\vx,\vy\vert\vz) \, d\vx \, d\vy = \mum, 
  \end{equation}
for all $\vz \in \Omega$, i.e. the unbinding rate is spatially constant. Note, however, as mentioned in the introduction Smoldyn's standard reaction model is based on the pure Smoluchowski~\cite{SmoluchowskiDiffLimRx} diffusion limited reaction mechanism, which corresponds to molecules reacting instantly upon reaching a separation given by the reaction-radius, $\rb$. It is therefore not a special case of the SCK CR model we now consider.

We now show what $\km(\vz)$ must be to satisfy detailed balance in the SCK model with the mirror-like reflection unbinding process, and demonstrate it is inconsistent with~\eqref{eq:SmoldynUnbindingProb}. In what follows, we assume $\Omega$ is convex so that binding events are always successful. If detailed balance holds, by~\eqref{eq:detailedBalancekpkm}, \eqref{eq:lambdaMultParticleCR}, and \eqref{eq:SmoldynCPlacement}, $\km(\vx,\vy\vert\vz)$ can be rewritten as
\begin{equation*}
  \begin{aligned}
  \km(\vx,\vy\vert\vz) &= \Kd \kp(\vz\vert\vx,\vy) \\
  &= \Kd\sqrt{2} \alpha \delta_{\dB_{\rb}(\vO)}(\vx-\vy)\delta(\vz - \gamma\vx - (1-\gamma)\vy).
  \end{aligned}
\end{equation*}
Then $\km(\vz)$ becomes
\begin{widetext}    
\begin{equation}
\label{eq:SmoldynUnbindingDB}
  \begin{aligned}
    \km(\vz) &= \Kd\sqrt{2} \alpha\int_{\Omega^2} \delta_{\dB_{\rb}(\vO)}(\vx-\vy)\delta(\vz - \gamma\vx - (1-\gamma)\vy) \, d\vx \, d\vy, \\
    &= \frac{\Kd\sqrt{2}\alpha}{(1-\gamma)^d}\int_{\Omega}\delta_{\dB_{\rb}(\vO)}\paren{\frac{\vx-\vz}{1-\gamma}}\ind_{\Omega}\paren{\frac{\vz - \gamma\vx}{1-\gamma}} \, d\vx, \\
    &= \frac{\Kd\sqrt{2}\alpha}{(1-\gamma)^{(d-1)}}\int_{\R^d} \delta_{\dB_{(1-\gamma)\rb}(\vz)}(\vx)
    \ind_{\Omega}\paren{\vx} \ind_{\Omega}\paren{\frac{\vz - \gamma \vx}{1-\gamma}} \, d\vx.
  \end{aligned}
\end{equation}
\end{widetext}
Here in the third equation we have used the identity that for a scaling constant, $\zeta > 0$, 
\begin{equation} \label{eq:deltasurfmeasidentity}
\delta_{\dB_{\rb}(\vO)}\paren{\frac{\vx-\vz}{\zeta}} = \zeta \delta_{\dB_{\zeta\rb}(\vz)}(\vx),
\end{equation}
which is shown in Appendix~\ref{app:proofOfDeltaSurfScaling}. Defining the translated and dilated set $\hat{\Omega}_{\gamma}(\vz)$ by
\begin{equation}\label{eq:omegahat}
\hat{\Omega}_{\gamma}(\vz) = \left\{ \vx \in \Omega \,\middle|\, \frac{\vz - \gamma \vx}{1 - \gamma} \in \Omega\right\},
\end{equation}
we have that~\eqref{eq:SmoldynUnbindingDB} simplifies to
\begin{equation*}
\km(\vz) = \frac{\Kd\sqrt{2}\alpha}{(1-\gamma)^{(d-1)}} \abs{\dB_{(1-\gamma)\rb}(\vz) \cap \hat{\Omega}_{\gamma}(\vz)}.
\end{equation*}
Near $\partial \Omega$ the intersection of the two sets will vary in $\vz$, implying that $\km(\vz)$ is again non-constant. With a spatially uniform unbinding rate, e.g.~\eqref{eq:SmoldynUnbindingProb}, we conclude the reflection kernel for handling domain boundaries does not preserve detailed balance at all points in $\Omega$. 

\subsubsection{The $\lambda-\rho$ Model}
The $\lambda-\rho$ model for reversible reactions in~\cite{ErbanChapman2011} considers the case that the \textrm{B} and \textrm{C} molecules stay stationary at the origin ($\DB = \DC = 0$ and $\vy = \vz = \vO$). Any \textrm{A} molecule is allowed to react with the \textrm{B} molecule with rate $\lambda$ when within $\rb$ of each other so that
  \begin{equation*}
    \kp(\vz\vert\vx,\vy) = \lambda\ind_{B_{\rb}(\vO)}(\vx)\delta(\vy)\delta(\vz).
  \end{equation*}
In one version of the model, upon unbinding a newly created \textrm{A} molecule is placed uniformly on a sphere with radius $\bar{\sigma}$ about the origin, with the stationary \textrm{B} molecule placed at the origin so that
  \begin{equation*}
    \km(\vx,\vy\vert\vz) = \mum\frac{\delta_{\partial B_{\bar{\sigma}}(\vO)}(\vx)}{\abs{\partial B_{\bar{\sigma}}(\vO)}}\delta(\vy)\delta(\vz).
  \end{equation*}

We note that the reaction kernels in the $\lambda-\rho$ model are in different functional forms -- the binding kernel is given by an indicator function on a ball, whereas the unbinding kernel is defined by a spherical $\delta$ surface measure. By \eqref{eq:detailedBalancekpkm} we can immediately conclude that these reaction kernels can not preserve detailed balance.


\subsection{Detailed Balance Leads to Rejection of Some Unbinding Events Near Boundaries}
\label{sect:dbreject}
We now show that for the standard form of the Doi model with detailed balance, near domain boundaries unbinding rates are effectively reduced. An analogous result is derived for the standard form of the SCK model in Appendix~\ref{sect:DBrejectionSCK}. These reduced unbinding rates can be realized in simulations through rejection of unbinding events that would lead to molecules being placed outside the domain. In contrast, once a reactant \textrm{C} molecule is sufficiently far from the boundaries that product molecules would always be placed within the domain, unbinding reactions are always successful.
\label{sect:DBrejection}
\begin{theorem}
  \label{thm:plcmProb}
  Consider the Doi binding kernel~\eqref{eq:fwdRxDoi} with corresponding detailed balance preserving Doi unbinding kernel~\eqref{eq:mTwoPartDefVR} and unbinding rate constant, $\mum$. For a C molecule at $\vz$ in $\Omega$,
  the probability per time the molecule successfully dissociates is
  \begin{equation} 
    \int_{\Omega^2}\km(\vx,\vy\vert\vz) \, d\vx \, d\vy < \mum
    \label{eq:DoiPlcmProb}
  \end{equation}
  when $\vz$ is within $\max(\gamma,1-\gamma)\rb$ of $\partial\Omega$. Similarly,
  \begin{equation} 
    \int_{\Omega^2}\km(\vx,\vy\vert\vz) \, d\vx \, d\vy = \mum
    \label{eq:DoiPlcmProb_away_from_bd}
  \end{equation}
  if $\vz$ is at least $\max(\gamma,1-\gamma)\rb$ away from $\partial\Omega$. 
\end{theorem}
\begin{proof}
  We will consider two cases: $\max(\gamma,1-\gamma) = 1-\gamma$ and $\max(\gamma,1-\gamma) = \gamma$. First assume $\max(\gamma,1-\gamma) = 1-\gamma$. The probability per time a \textrm{C} molecule at $\vz$ will successfully produce \textrm{A} and \textrm{B} molecules each within $\Omega$ is given by
  \begin{equation}
    \begin{aligned}
      \km(\vz) &= \int_{\Omega^2}\km(\vx,\vy\vert\vz) \, d\vx \, d\vy,\\  
      &= \mum\int_{\Omega}\frac{\ind_{B_{(1-\gamma)\rb}(\vz)}(\vx)}{\abs{B_{(1-\gamma)\rb}(\vO)}}\ind_{\Omega}\paren{\frac{\vz-\gamma\vx}{1-\gamma}} \, d\vx,\\ 
      &= \mum\frac{\abs{B_{(1-\gamma)\rb}(\vz)\cap\hat{\Omega}_{\gamma}(\vz)}}{\abs{B_{(1-\gamma)\rb}(\vO)}},
    \end{aligned}
    \label{eq:DoiUnbindProbSuccess}
  \end{equation}
where $\hat{\Omega}_{\gamma}(\vz)$ denotes the translated and dilated set defined in~\eqref{eq:omegahat}. Let $d(\vz,\partial\Omega)$ be the Euclidean distance from $\vz$ to $\partial\Omega$, and assume $d(\vz,\partial\Omega) < (1-\gamma)\rb$ so that 
  \begin{equation*}
    \abs{B_{(1-\gamma)\rb}(\vz)\cap\Omega} < \abs{B_{(1-\gamma)\rb}(\vO)}.
  \end{equation*}
  By \eqref{eq:DoiUnbindProbSuccess} this immediately gives
  \begin{equation*}
    \km(\vz) < \mum.
  \end{equation*}

Now consider $\vz$ with $d(\vz,\partial\Omega)\geq(1-\gamma)\rb$ so that the ball of radius $(1-\gamma)\rb$ about $\vz$ is strictly contained in $\Omega$.  Let $\vw \in B_{(1-\gamma)\rb}(\vz)$. We claim $\vw \in \hat{\Omega}_{\gamma}(\vz)$. Take
\begin{equation}
  \vv = \tfrac{1}{1-\gamma} \paren{\vz-\gamma \vw}
  \label{eq:OmegahatPt}
\end{equation}
so that 
\begin{equation*}
  \abs{\vv - \vz}  
  = \frac{\gamma}{1-\gamma}\abs{\vz-\vw} < \gamma\rb.
\end{equation*}
This implies that $\vv \in B_{\gamma \rb}(\vz)$, and hence in $\Omega$ as $B_{\gamma \rb}(\vz) \subset B_{(1-\gamma)\rb}(\vz) \subset \Omega$. By definition~\eqref{eq:omegahat} this shows that $\vw \in \hat{\Omega}_{\gamma}(\vz)$, which implies $B_{(1-\gamma)\rb}(\vz) \subset \hat{\Omega}_{\gamma}(\vz)$. By~\eqref{eq:DoiUnbindProbSuccess}
\begin{equation*}
  \km(\vz) = \mum\frac{\abs{B_{(1-\gamma)\rb}(\vz)}}{\abs{B_{(1-\gamma)\rb}(\vO)}} = \mum.
\end{equation*}
  
Finally, we note that using scaling properties of the $\delta$-function and the indicator function, we can equivalently write
\begin{align*}  
 \km(\vx,\vy\vert\vz) &= \mum \tfrac{(1-\gamma)^d}{\gamma^d} \delta \paren{\vx - \frac{\vz - (1-\gamma)\vy}{\gamma}} 
 \frac{\ind_{B_{(1-\gamma)\rb}(\vO)}\paren{\tfrac{1-\gamma}{\gamma}(\vy-\vz)}}{\abs{B_{(1-\gamma)\rb}(\vO)}} \\
 &= \mum \delta \paren{\vx - \frac{\vz - (1-\gamma)\vy}{\gamma}} \frac{\ind_{B_{\gamma\rb}(\vO)}\paren{\vy-\vz}}{\abs{B_{\gamma\rb}(\vO)}}. \\
\end{align*}
Integrating first with respect to $\vx$, we obtain
  \begin{equation} \label{eq:DoiUnbindProbSuccess2}
    \begin{aligned}
      \km(\vz) &= \mum\int_{\Omega}\frac{\ind_{B_{\gamma\rb}(\vz)}(\vy)}{\abs{B_{\gamma\rb}(\vO)}}\ind_{\Omega}\paren{\frac{\vz-(1-\gamma)\vy}{\gamma}} \, d\vy,\\ 
      &= \mum\frac{\abs{ B_{\gamma\rb}(\vz)\cap\hat{\Omega}_{1-\gamma}(\vz)}}{\abs{B_{\gamma\rb}(\vO)}}.
    \end{aligned}
  \end{equation}
A similar argument to above interchanging $\gamma$ and $1-\gamma$ then gives the result when $\max(\gamma,1-\gamma) = \gamma$.
\end{proof}

In simulations we can incorporate the spatially varying dissociation rate near $\partial \Omega$ given by~\eqref{eq:DoiUnbindProbSuccess} in several ways. One could pre-tabulate $\km(\vz)$, or dynamically calculate $\km(\vz)$ based on the position of a \textrm{C} molecule. Alternatively, one could use the previously mentioned rejection mechanism by interpreting~\eqref{eq:DoiUnbindProbSuccess} as the probability per time a \textrm{C} molecule attempts to dissociate apart, $\mum$, multiplied by the probability such a dissociation event is successful. In this case we could first sample the position, $\vx$, of the \textrm{A} molecule uniformly within a ball of radius $(1-\gamma)$, and then place the \textrm{B} molecule by reflection across the sphere at a position $\vy$ such that $\vz = \gamma \vx + (1-\gamma) \vy$. This method is illustrated in Fig.~\ref{fig:Fig2}A and Fig.~\ref{fig:Fig2}B. The probability density for these placement steps is just
\begin{equation*}
\delta \paren{\vy - \frac{\vz - \gamma \vx}{1 - \gamma}} \frac{\ind_{B_{(1-\gamma)\rb}(\vz)}(\vx)}{\abs{B_{(1-\gamma)\rb}(\vO)}},
\end{equation*}
which gives the corresponding (acceptance) probability that the positions of the two molecules are within $\Omega$, 
\begin{align*}
\prob \brac{\vx \in \Omega, \vy \in \Omega \,\vert\, \vz} 
&= \frac{\abs{B_{(1-\gamma)\rb}(\vz)\cap\hat{\Omega}_{\gamma}(\vz)}}{\abs{B_{(1-\gamma)\rb}(\vO)}},
\end{align*}
as derived in ~\eqref{eq:DoiUnbindProbSuccess}. Note, one could equivalently sample the \textrm{B} molecule's position uniformly within the sphere of radius $\gamma \rb$ about $\vz$, and then choose the \textrm{A} particle's position by reflection, corresponding to~\eqref{eq:DoiUnbindProbSuccess2}. 

\subsection{Preserving a constant unbinding rate} \label{sect:kpfromkm}
Our results so far show that assuming a forward reaction kernel, $\kp\paren{\vz \vert \vx,\vy}$, and choosing $\km$ via the detailed balance relation~\eqref{eq:detailedBalanceRxEq} results in a spatially varying unbinding rate near the domain boundary. We now investigate what happens if we formulate $\km\paren{\vx,\vy\vert\vz}$ to have a uniform unbinding rate, $\km(\vz) = \mum$ at all points within $\Omega$, and then choose the association kernel, $\kp$, via~\eqref{eq:detailedBalanceRxEq}.
  
For simplicity, we will restrict attention to a version of the Doi model with a simplified placement kernel for the products of the $\textrm{C} \to \textrm{A} + \textrm{B}$ reaction, but expect our analysis could be adapted to more general placement kernels and the SCK model too. We take $\gamma = 0$ in the standard Doi model, and normalize the placement kernel to guarantee particles are always successfully placed within the domain:
\begin{equation*}
  \km\paren{\vx,\vy\vert\vz} = \mum \frac{\ind_{B_{\rb}(\vO)}(\vx-\vz)}{\abs{B_{\rb}(\vz) \cap \Omega}} \delta \paren{\vy - \vz}.
\end{equation*}
This model corresponds to placing the $\vx$ particle uniformly within the portion of the ball of radius $\rb$ about $\vz$ that is within the domain, and then placing the $\vy$ particle at the position of the $\vz$ particle. Sufficiently far from boundaries it is consistent with the kernels used in~\cite{Donevetal2018}. Note that 
\begin{equation*}
  \km(\vz) = \int_{\Omega^2} \km\paren{\vx,\vy\vert\vz} \, d\vx \, d\vy = \mum
\end{equation*}
with these choices so that unbinding is always successful.

Applying~\eqref{eq:detailedBalanceRxEq}, we then have that 
\begin{equation*}
  \kp\paren{\vz\vert\vx,\vy} = \frac{\mum}{\Kd} \frac{\ind_{B_{\rb}(\vO)}(\vx-\vz)}{\abs{B_{\rb}(\vz) \cap \Omega}} \delta \paren{\vy - \vz},
\end{equation*}
so that the probability per time an $\textrm{A}$ and $\textrm{B}$ molecule can react is
\begin{align*}
  \kp(\vx,\vy) &= \frac{\mum}{\Kd} \int_{\Omega} \frac{\ind_{B_{\rb}(\vO)}(\vx-\vz)}{\abs{B_{\rb}(\vz) \cap \Omega}} \delta \paren{\vy - \vz} \, d\vz \\
  &= \frac{\mum}{\Kd} \frac{\ind_{B_{\rb}(\vO)}(\vx-\vy)}{\abs{B_{\rb}(\vy) \cap \Omega}}.
\end{align*}
Here the effective probability per time two particles can react when separated by $\rb$ or less is $\mum / (\Kd \abs{B_{\rb}(\vy) \cap \Omega})$, which will \emph{increase} for $\vy$ sufficiently close to the boundary.

It is therefore possible to achieve a uniform unbinding rate using a standard product placement model, however, we see that we then require an increased pointwise association rate in the vicinity of the domain boundary for the underlying particle model to be consistent with detailed balance holding.

\subsection{Numerical Examples} \label{sect:numerics}
As a simple illustration of how preserving detailed balance near domain boundaries can impact model predictions, we now consider Brownian Dynamics (BD) simulations of the standard Doi VR model of Section~\ref{sect:VR_DB}, using both our rejection mechanism that is consistent with detailed balance holding, and the mirror-like reflection mechanism of Fig.~\ref{fig:Fig2}C. Our basic BD algorithm is summarized in Algorithm~\ref{alg:BD} of Appendix~\ref{sect:numericsappendix}.

\begin{table}[tb] 
\caption{Parameters for Brownian Dynamics (BD) Simulations} 
\label{tab:tab1} 
\begin{center} 
  \begin{tabular}{c|lc} 
    \hline 
    Parameter & Description & Value\\ \hline
    $D$           & diffusion coefficient & 1 $\mu$m$^2$ s$^{-1}$\\ 
    $\lambda$     & association rate & varies, s$^{-1}$,\\
    & &see Appendix~\ref{sect:numericsappendix}\\
    $\rb$          & reaction-radius & $10$ nm \\
    $\mum$        & dissociation rate & $17.3$ s$^{-1}$ \\
    $\brac{C_0}$  & initial \textrm{C} concentration  & $1.25 \times 10^{-5}$ nm$^{-3}$\\ 
    $\brac{A_0}, \brac{B_0}$ & initial \textrm{A} and \textrm{B} concentrations   & 0 nm$^{-3}$\\ 
    $L$       & domain length & varies, nm \\
    $\Omega$ & domain & $(0,L)^3$ or \\
    & &$L\times L \times 30 \text{ nm}$ \\
    $C_{\textrm{max}}$ & maximum number of \textrm{C} particles & $\brac{C_0} \abs{\Omega}$ \\
    $\Delta t$ & BD step size & $10^{-8}$ s \\
    $T$      & final time & $.5$ s
\end{tabular}
\end{center} \end{table}

We consider the reversible $\textrm{A} + \textrm{B} \leftrightarrows \textrm{C}$ reaction in three cubic domains of increasing size, $\Omega = (0,L)^3$ with $L \in \{20 \text{ nm}, 100 \text{ nm}, 200 \text{ nm}\}$. Particles were initialized in the bound state, i.e. as \textrm{C} particles, with a fixed initial concentration of $\brac{C_0} = 1.25 \times 10^{-5} (\text{nm})^{-3}$ that was initially uniformly distributed across $\Omega$. This corresponded to $C_{\textrm{max}} = 1$ particle for $L = 20 \text{ nm}$, $C_{\textrm{max}} = 125$ particles for $L = 100 \text{ nm}$, and $C_{\textrm{max}} = 1000$ particles for $L = 200 \text{ nm}$. All boundaries were treated as reflecting.

For the smallest domain we chose $\lambda = 40.5745 \text{s}^{-1}$ so that by~\eqref{eq:ssprobs} the steady-state probability of being in the bound state was $0.5512$ when using the rejection model. Let $\avg{C(t)}$ denote the average number of \textrm{C} molecules within the system at time $t$, with 
\begin{equation} \label{eq:Cfrac}
  f(t) = \frac{\avg{C(t)}}{C_{\textrm{max}}}
\end{equation}
denoting the average fraction of bound \textrm{A} (or \textrm{B}) particles in the system at $t$. As the domain size was increased, $\lambda$ was increased so that the steady state value of $f(t)$ in a corresponding well-mixed chemical master equation (CME) model was held constant at $.551$ (i.e. three digits of accuracy). In the well-mixed CME model the association rate for the reaction was given by the fast diffusion limit~\cite{IsaacsonCRDME2013,ErbanChapman2009} $\lambda \tfrac{4}{3}\pi \rb^3$, with a dissociation rate of $\mum$. Our method for calculating $\lambda$ from the steady-state value of $\lim_{t \to \infty} f(t) = .551$ in the well-mixed CME is summarized in Appendix~\ref{sect:numericsappendix}. All simulations used a timestep of $\Delta t = 10^{-8} \text{ s}$, which is also discussed in Appendix~\ref{sect:numericsappendix}. All other parameters are given in Table~\ref{tab:tab1}.

\begin{figure*}
  \includegraphics[width=\textwidth]{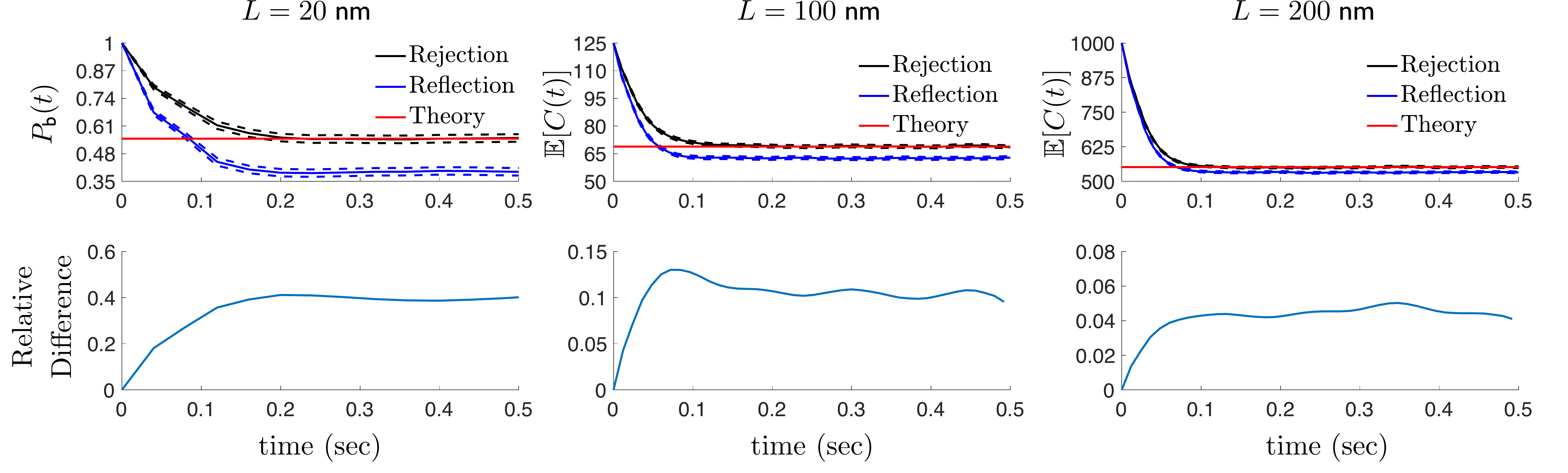}
  \caption{How the average number of \textrm{C} particles changes when
  using the detailed balance preserving rejection method versus the reflection
  method for dissociation reactions. See Section~\ref{sect:numerics} for
  details, and both Table~\ref{tab:tab1} and
  Appendix~\ref{sect:numericsappendix} for parameter values. The domains were
  taken to be cubes of width $L$ with reflecting boundaries. Initially
  $C_{\textrm{max}}$ particles of species \textrm{C} were distributed uniformly
  throughout the domain. Each curve is an average from 11,000 Brownian Dynamics
  simulations of the $\textrm{A} + \textrm{B} \leftrightarrows \textrm{C}$
  reaction using the method described in Appendix~\ref{sect:numericsappendix}.
  Dashed lines give 95\% confidence intervals. Solid red lines represent the
  analytical value for $\bar{P}_{\textrm{b}}$ (upper left panel)
  from~\eqref{eq:ssprobs}, and the steady-state value of $\avg{C(t)}$ from a
  well-mixed chemical master equation (CME) model (upper middle and right
  panels). Both $\bar{P}_{\textrm{b}}$ and the CME steady-state for
  $\avg{C(t)}/C_{\textrm{max}}$ were fixed at $.551$, with the Doi association
  probability per time, $\lambda$, chosen to give this steady-state value in the
  well-mixed CME model, see Appendix~\ref{sect:numericsappendix}.}
  \label{fig:Fig3}
\end{figure*}

In the upper panels of Figure~\ref{fig:Fig3} we compare the average number of \textrm{C} molecules in the system, $\avg{C(t)}$, using both the detailed balance preserving rejection method and the reflection method to handle \textrm{A} or \textrm{B} particles that end up outside the domain when the $\textrm{C} \to \textrm{A} + \textrm{B}$ reaction occurs. In the figures the curve labelled "Theory" refers to the corresponding value from the well-mixed CME model to which $\lambda$ was calibrated, i.e. $.551 C_{\textrm{max}}$. Note that $P_{\textrm{b}}(t) = \avg{C(t)}$ in the case that $C_{\textrm{max}}=1$ (i.e. when $L=20 \text{ nm}$). Here we used the same parameters for each method, and estimated $\avg{C(t)}$ by averaging over 11,000 simulations. We see clear differences between the steady-state value of $\avg{C(t)}$ obtained by the reflection method (solid blue line) and the chemical master equation value (solid red line) for small domain sizes, while the detailed balance preserving rejection method (solid black line) always matches the well-mixed steady-state value. As the simulation domain size is increased, thereby reducing the surface to volume ratio of the domain and frequency of dissociation reactions occurring near boundaries, the discrepancy between the reflection method and the CME steady-state value decreases. The bottom row of panels in Fig.~\ref{fig:Fig3} show the relative difference of the curves obtained by the reflection method relative to that of the rejection method.

\begin{figure}[tb]
  \includegraphics[width=.4\textwidth]{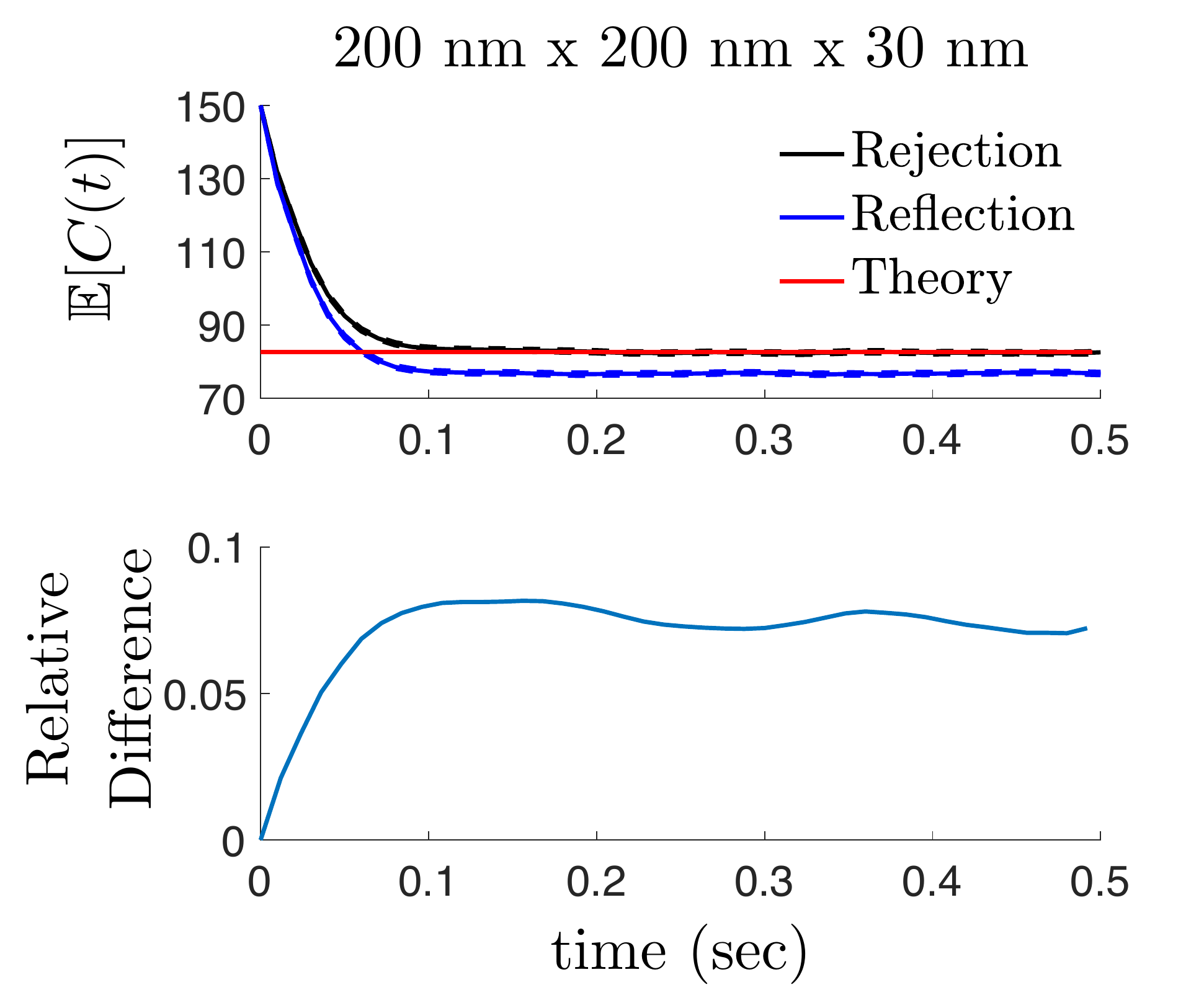}
  \caption{How the average number of \textrm{C} particles changes when using the detailed balance preserving rejection method versus the reflection method for dissociation reactions in a box of size $200 \times 200 \times 30$ nm. $C_{\textrm{max}} = 150$ particles of type \textrm{C} were initially uniformly distributed within the domain. All other parameters are the same as described in Fig.~\ref{fig:Fig3}.}
  \label{fig:Fig4}
\end{figure}

How strongly the reflection-based approach disagrees from the rejection-based approach is a complicated function of reaction and transport parameters, particle densities, possible reactions, and domain geometry. For example, in Fig.~\ref{fig:Fig4} we consider the same comparison as in Fig.~\ref{fig:Fig3}, but in a domain of size $200 \text{ nm} \times 200 \text{ nm}\times 30\text{ nm}$, where the width of $30 \text{ nm}$ is comparable to the width of the interior of yeast endoplasmic reticulum sheets~\cite{ERThickness2010}. We see a clear increase in the relative difference between the rejection and reflection approaches compared to the cubical domain of width $200 \text{ nm}$. 

\section{Discussion}
\begin{figure*}[tb]
  \includegraphics[width=.9\textwidth]{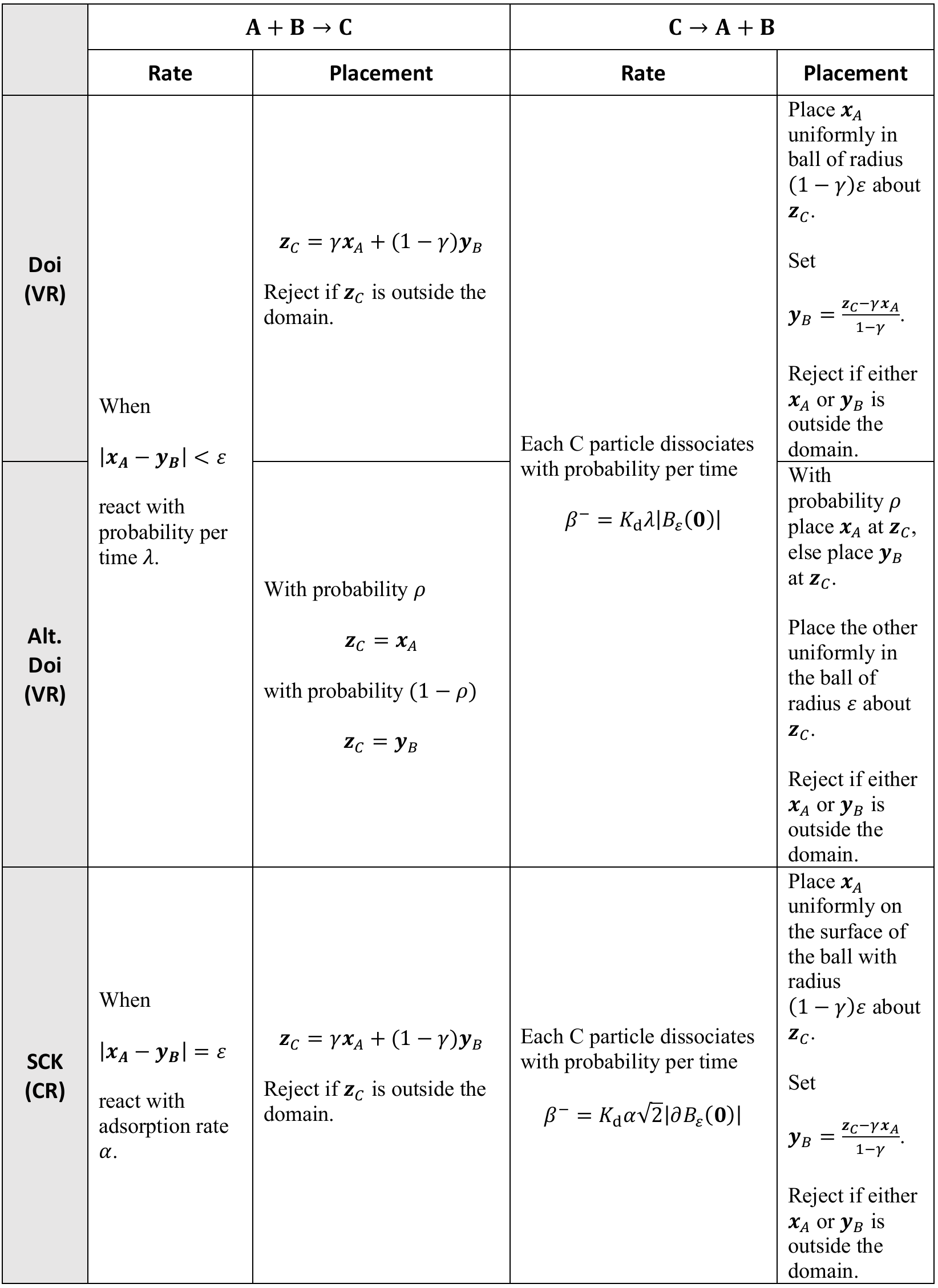}
  \caption{Table of Doi (VR) and SCK (CR) variants that are consistent with detailed balance holding. Here $\vx_{\textrm{A}}$, $\vy_{\textrm{B}}$, and $\vz_{\textrm{C}}$ denote the positions of the \textrm{A}, \textrm{B}, and \textrm{C} particles involved in the reaction as reactants and/or products.}
  \label{fig:Fig5}
\end{figure*}
Methods for preserving detailed balance in simulations of particle-based stochastic reaction-diffusion models have been considered recently in a variety of contexts~\cite{MorellitenWold2008,Donevetal2018,FrohnerNoe2018}. These works focus on reactions in freespace or periodic domains, raising the question of how such schemes should work for reactions occurring near (reflecting) domain boundaries. For the general volume reactivity and contact reactivity models in a closed domain with reflecting boundaries, we have formulated a pointwise detailed balance condition for the reversible $\textrm{A} + \textrm{B} \leftrightarrows \textrm{C}$ reaction, and illustrated how enforcing detailed balance offers guidelines for the placement of reaction products in simulations. 
 
When using common variants of the Doi VR and Smoluchowski-Collins-Kimball CR models in closed domains, it is a modeling choice how reaction products should be placed for reactants close to the boundaries. We demonstrated that for common choices of association kernels, preserving pointwise detailed balance requires a decreased and spatially varying unbinding / dissociation rate near domain boundaries (even in convex domains). Our work provides one simple approach to realize this rate in typical Brownian Dynamics and lattice simulation algorithms; initially ignoring the domain boundary when placing reaction product(s), and then rejecting unbinding events where one of the products was placed outside the domain. One benefit to this approach is that no modification to the underlying reaction kernels used for the binding and unbinding process is needed for common variants of the Doi and SCK models, or for recently proposed versions that have been shown to satisfy detailed balance in periodic domains~\cite{Donevetal2018}. Note, some lattice jump-process simulation methods can trivially enforce (pointwise) detailed balance of spatial reaction fluxes for such VR kernels by appropriate choice of dissociation (or association) transition rates~\cite{IsaacsonZhang17}. In Figure~\ref{fig:Fig5} we summarize several options for standard Doi (VR) and SCK (CR) reaction rate and product placement models which are consistent with detailed balance holding in a general bounded domain with reflecting boundaries.

We note that other approaches for handling product placement near domain boundaries, such as the boundary reflection model, can preserve a spatially-uniform unbinding rate for the $\textrm{C} \to \textrm{A} + \textrm{B}$ reaction. However, as we showed this leads to models that violate detailed balance of (pointwise) spatial reaction fluxes in the neighborhood of domain boundaries when using a standard association kernel. We therefore formulated a specific variant of the Doi model unbinding kernel that includes a spatially-uniform unbinding rate, and derived the corresponding detailed balance preserving association kernel. We found that to compensate for the uniform unbinding rate, the effective probability per time two reactants within a reaction-radius can react must be increased in the vicinity of the domain boundary. Though we do not show it here, we expect that a similar increase would occur when modifying other commonly-used Doi and SCK unbinding kernels to support a spatially-uniform unbinding rate. 

\textit{Importance and benefits of preserving detailed balance:}
As mentioned in the introduction, we expect detailed balance of (pointwise) spatial reaction fluxes to hold at all points within the domain from both time reversibility of more microscopic models, and from statistical mechanics arguments~\cite{Donevetal2018}. Analogous to how we  want reaction models to preserve mass at a population level in the $\textrm{A} + \textrm{B} \leftrightarrow \textrm{C}$ reaction, it seems desirable to have PBSRD models with reaction kernels that preserve detailed balance. In addition, as illustrated by~\eqref{eq:VR_Kd}, when detailed balance is known to be satisfied, equilibrium dissociation constants can be used in estimating (microscopic) reaction parameters, eliminating the need to directly parametrize one of the microscopic association or dissociation rates (which can be more difficult to measure). Preserving detailed balance at the particle level should also ensure that it is preserved in more macroscopic models that can be derived from PBSRD models. For example, the corresponding large-population limit of the VR model gives macroscopic mean-field partial integral differential equation (PIDE) models for deterministic concentration fields~\cite{IsaacsonMaSpiliopoulosSIMA22,IsaacsonMaSpiliopoulosSIAP22}. We would expect detailed balance being satisfied at the PBSRD level to imply it is also satisfied in the derived PIDEs.

Whether preserving detailed balance near domain boundaries will significantly impact model predictions is likely to be a model-specific question that depends on reaction parameters, transport parameters, domain shape, and the concentration of chemical species. While it is beyond the scope of the current work to systematically explore all these different degrees of freedom, we can make some general statements. If one is interested in dynamics near the domain boundary, for example interactions between cytosolic proteins and membrane bound/tethered proteins or receptor tails~\cite{Zhang2019}, then how boundary interactions are handled may become important in model predictions. Likewise, when trying to resolve chemical dynamics within a tortuous space like the inside of the cystosol or endoplasmic reticulum, where boundary surfaces can block significant portions of the domain, it may be that the choice of how reaction kernels handle boundary interactions can significantly impact PBSRD model predictions. 

Similarly, our numerical examples show that if one is interested in reaction dynamics within a domain that is comparable in size to interaction distances, then boundary effects would be expected to play a role in model predictions, potentially giving altered dissociation/equilibrium constants in methods that do not preserve detailed balance. In contrast, as demonstrated by our larger cubical domain example, boundary effects are less important when studying reaction processes within large, open regions with lower surface to volume ratios. We believe it would be an interesting future study to investigate the impact of preserving versus violating pointwise detailed balance near boundaries for varying choices of reaction kernels, physical parameters (reaction rates, diffusivities, interaction distances), species concentrations, reaction networks, reaction localizations, and domain geometries.

Finally, we note that the relative importance of preserving a constant unbinding rate, preserving a constant association rate for molecules that are appropriately separated, and preserving detailed balance at all points in a domain, may depend on the underlying physical process being approximated. If preserving all three is desirable for a given model, alternative choices for product placement kernels in association or dissociation reactions appear to be needed.

\textit{Additional Future work:}
Here we have only discussed detailed balance conditions for (point) particle-based stochastic reaction-diffusion models, focusing on a perspective of adapting unbinding kernels to preserve detailed balance for specified functional forms of association kernels. An interesting future direction would be to more carefully consider families of unbinding kernels that allow for a uniform unbinding rate, and investigate if any give corresponding detailed balance-preserving association kernels in which the probability per time sufficiently close reactants can react is spatially-uniform.

This work also only analytically studied the two-particle $\textrm{A} + \textrm{B} \leftrightarrows \textrm{C}$ reaction; it would be helpful to analytically confirm our expectation that our results should be fully consistent with detailed balance holding in the multi-particle case. In addition, we treated molecules as point particles, ignoring other physical effects that can be important in some contexts. These include volume exclusion and other finite size effects that may be important in more dense systems~\cite{KleinSchwarz2014}, along with potential interactions between particles~\cite{FrohnerNoe2018}. We hope to report on more general multiparticle systems in future work.

Finally, we note that it has recently been demonstrated that PBSRD models predict reduced macroscopic reaction rates for molecules localized near surface boundaries~\cite{AndrewsSurfEffects20}. This suggests an analogous future problem to study; whether more microscopic Langevin dynamics or molecular dynamics models might suggest explicit modifications to common PBSRD reaction kernels near boundary surfaces.

\section{Acknowledgments}
Both authors' work was supported by National Science Foundation award DMS--1902854. SAI was partially supported by a grant from the Simons Foundation, and thanks the Isaac Newton Institute of Mathematical Sciences for hosting him as a visiting Simons Fellow for the program on Stochastic Dynamical Systems in Biology during a portion of the period when this work was carried out. SAI thanks Aleksander Donev, David Isaacson, Peter Kramer, Charles S. Peskin, and Konstantinos Spiliopoulos for helpful discussions about detailed balance in stochastic; and classical, quantum, and statistical mechanical systems. 

\appendix

\section{PROOF OF (\ref{eq:lambdaMultParticleCR})}
\label{app:proofOfEqCR}
In this section, we prove \eqref{eq:lambdaMultParticleCR}.
\begin{equation*}
  \delta_{\dRset}(\vx,\vy) = \sqrt{2}\delta_{\partial B_{\rb}(\vO)}(\vx - \vy).
\end{equation*}
\begin{proof}
  Let $v(\vx,\vy)$ be an arbitrary test function with compact support. We define a change of variables
  \begin{equation}
      \vR = \vx - \vy, \,\,\,\,\,\,\,\,\,\, \vw = \vx + \vy.
    \label{eq:cov}
  \end{equation} 
  To prove \eqref{eq:lambdaMultParticleCR}, we subsequently prove
  \begin{equation}
    \int_{\dRset}v(\vx, \vy) \, dS(\vx, \vy) = \sqrt{2}\int_{\R^d}\int_{\partial B_{\rb}(\vO)}v(\vx(\vR, \vw), \vy(\vR, \vw))\, dS(\vR) \, d\vw.
    \label{eq:IdToProof}
  \end{equation}
  
  By the definition of the delta surface measure, we begin by rewriting the integral as
   \begin{equation}
       \int_{\R^{2d}}v(\vx, \vy)\delta_{\dRset}(\vx,\vy) \, d\vx \, d\vy = \int_{\dRset}v(\vx, \vy) \, dS(\vx, \vy).
       \label{eq:defnDeltaSurf}
   \end{equation}
  We let $\phi(\vx, \vy) = \abs{\vx-\vy}$, which gives $\abs{\nabla\phi(\vx,\vy)} = \sqrt{2}$. Applying the co-area formula we obtain
    \begin{equation}
      \begin{aligned}
        \int_{\dRset}&v(\vx,\vy) \, dS(\vx,\vy) \\
        &= \int_{\{\abs{\vx-\vy} = \rb\}}v(\vx,\vy) \, dS(\vx,\vy), \\
        &= \int_{\phi^{-1}(\rb)}v(\vx,\vy) \, dS(\vx,\vy),\\
        &= \int_{\R^{2d}}v(\vx,\vy)\delta(\phi(\vx,\vy)-\rb)\abs{\nabla_{\vx,\vy}\phi(\vx,\vy)} \, d\vx \, d\vy, \\
        &= \sqrt{2}\int_{\R^{2d}}v(\vx,\vy)\delta(\phi(\vx,\vy)-\rb) \, d\vx \, d\vy.
      \end{aligned}
      \label{eq:coareaI}
    \end{equation}
  By a change of variables \eqref{eq:cov} we can rewrite the last integral in \eqref{eq:coareaI} as
    \begin{equation}
      \int_{\R^{2d}}v(\vx,\vy)\delta(\phi(\vx,\vy)-\rb) \, d\vx \, d\vy = \frac{1}{2^d}\int_{\R^{2d}}v(\vx(\vR, \vw), \vy(\vR, \vw))\delta\paren{\phi(\vx(\vR, \vw), \vy(\vR, \vw))-\rb} \, d\vR \, d\vw,
      \label{eq:covIntEqn}
    \end{equation}
    where $\phi(\vx(\vR, \vw), \vy(\vR, \vw)) = \phi(\vR) = \abs{\vR}$, which gives $\abs{\nabla_{\vR}\phi(\vR)} = 1$. Applying the co-area formula again we have
    \begin{equation}
      \begin{aligned}
        \int_{\R^{2d}}&v(\vx(\vR, \vw), \vy(\vR, \vw))\delta\paren{\phi(\vx(\vR, \vw), \vy(\vR, \vw))-\rb} \, d\vR \, d\vw \\
        &= \int_{\R^d}\int_{\phi^{-1}(\rb)}\frac{v(\vx(\vR, \vw), \vy(\vR, \vw))}{\abs{\nabla_{\vR}\phi(\vR)}} \, dS(\vR) \, d\vw, \\
        &= \int_{\R^d}\int_{\partial B_{\rb}(\vO)}v(\vx(\vR, \vw), \vy(\vR, \vw))\, dS(\vR) \, d\vw \\
        &= \int_{\R^{2d}} \delta_{\partial B_{\rb}(\vO)}(\vR) v(\vx(\vR, \vw), \vy(\vR, \vw))\, d\vR \, d\vw \\
        &= 2^{d} \int_{\R^{2d}} \delta_{\partial B_{\rb}(\vO)}(\vx - \vy) v(\vx, \vy)\, d\vx \, d\vy.
      \end{aligned}
      \label{eq:coareaII}
    \end{equation}
  Equation~\eqref{eq:coareaI} together with \eqref{eq:covIntEqn} and \eqref{eq:coareaII} imply that
    \begin{equation*}
      \int_{\dRset}v(\vx,\vy) \, dS(\vx,\vy) = \sqrt{2}\int_{\R^{2d}} \delta_{\partial B_{\rb}(\vO)}(\vx - \vy) v(\vx, \vy)\, d\vx \, d\vy,
    \end{equation*}
  so that we have
  \begin{equation*}
    \delta_{\dRset}(\vx, \vy) = \sqrt{2}\delta_{\partial B_{\rb}(\vO)}(\vx - \vy).
  \end{equation*}
\end{proof}



\section{PROOF OF (\ref{eq:deltasurfmeasidentity})} \label{app:proofOfDeltaSurfScaling}
In this appendix we show that as a function of $\vx$ the following identity holds
\begin{equation} \label{eq:deltasurfmeasidentityAp}
\delta_{\dB_{\rb}(\vO)}\paren{\frac{\vx-\vz}{\zeta}} = \zeta \delta_{\dB_{\zeta\rb}(\vz)}(\vx)
\end{equation}
for $\zeta > 0$.
\begin{proof}
Using the change of variables
\begin{align*}
\vy = \frac{\vx - \vz}{\zeta},
\end{align*}
the action of the surface delta-function against a test function $\phi(\vx)$ is then
\begin{align*}
\int_{\R^d} \delta_{\dB_{\rb}(\vO)}\paren{\frac{\vx-\vz}{\zeta}}  &\phi(\vx) \, d\vx \\
&= \zeta^{d} \int_{\R^d} \delta_{\dB_{\rb}(\vO)} (\vy) \phi(\zeta \vy + \vz) \, d\vy \\
&= \zeta^{d} \int_{\dB_{\rb}(\vO)} \phi(\zeta \vy + \vz) dS(\vy) \\
&= \zeta^{d} \int_{\dB_{\rb}(\vz)} \phi(\zeta \vy) dS(\vy) \\
&= \zeta \int_{\dB_{\zeta \rb}(\vz)} \phi(\vx) dS(\vx) \\
&= \zeta \int_{\R^{d}} \delta_{\dB_{\zeta\rb}(\vz)}(\vx) \phi(\vx) \, d\vx.
\end{align*}
\end{proof}
Note, an immediate corollary is that
\begin{equation} \label{eq:deltasurfmeasidentityAp2}
\delta_{\dB_{\rb}(\vO)}\paren{\frac{\vx-\vz}{\zeta}} = \zeta \delta_{\dB_{\zeta\rb}(\vO)}(\vx-\vz).
\end{equation}


\section{Equivalence of Weak Forms of the SCK CR Model}
\label{sect:WeakFormContactRx}
In this appendix we show that our generalized model with $\delta$ surface-measure coefficients is consistent in weak form~\cite{EvansPDEs2ed,SchussStochProcBook2010} with the standard representation of the SCK CR model as a system of PDEs with reactive boundary conditions. 

We begin by considering the probability density of the unbound state in our formulation of the SCK CR model
\begin{equation}
\begin{aligned} 
   \ind_{\Omega^2\setminus\Rset\cup\dRset}(\vx,\vy)\PD{p}{t}(\vx,\vy,t) 
  &= \nabla_{\vx,\vy} \cdot \ind_{\Omega^2\setminus\Rset\cup\dRset}(\vx,\vy)\Dmat \nabla_{\vx,\vy} p(\vx,\vy,t) \\
  &\phantom{==} - \alpha \delta_{\dRset}(\vx,\vy) \ind_{\Omega}(\gamma\vx+(1-\gamma)\vy) p(\vx,\vy,t) \\
  &\phantom{==} + \mum\int_{\Omega} \frac{\delta_{\partial B_{(1-\gamma)\epsilon}(\vO)}(\vx-\vz)}{\abs{\partial B_{(1-\gamma)\rb}(\vO)}}\delta\paren{\vy - \frac{\vz-\gamma\vx}{1-\gamma}} \pb(\vz,t) \, d\vz,
\end{aligned}
\label{eq:Smoluchowskirho}
\end{equation}
with the reflecting zero Neumann boundary condition
\begin{equation}
    \nabla_{\vx,\vy}p(\vx,\vy,t)\cdot\veta(\vx,\vy) = 0, \quad (\vx,\vy) \in \partial (\Omegaf^2).
    \label{eq:SCKCR_NoFluxBC}
\end{equation} 
Here $\partial (\Omegaf^2)$ corresponds to the portion of $\partial \Omega \times \partial \Omega$ that is outside $\mathcal{R} \cup \partial \mathcal{R}$, and $\veta(\vx,\vy)$ denotes the unit outward normal vector to $\partial (\Omegaf^2)$ at $(\vx,\vy)$. We assume the probability density the two particles are initially unbound is zero within the reaction surface and its interior, ensuring the molecules never approach closer than $\varepsilon$. That is, we assume  $p(\vx,\vy,0) = 0$ for $(\vx,\vy) \in \Omega^{2} \cap (\Rset \cup \partial \Rset)$.

Denote by $\mathcal{V}$ a space of test functions, with $v(\vx,\vy) \in \mathcal{V}$. To obtain the weak form of~\eqref{eq:Smoluchowskirho}, we multiply by $v(\vx,\vy)$ and integrate both sides over $(\vx, \vy) \in \Omega^2$. We now simplify each term above one by one. The time-derivative term of~\eqref{eq:Smoluchowskirho} becomes
\begin{equation}
  \int_{\Omega^2}\ind_{\Omega^2\setminus\Rset\cup\dRset}(\vx,\vy) \PD{p}{t}(\vx,\vy,t)v(\vx,\vy) \, d\vx \, d\vy = \int_{\Omega^2\setminus\Rset\cup\dRset}\PD{p}{t}(\vx,\vy,t)v(\vx,\vy) \, d\vx \, d\vy.
  \label{eq:rhoLHS}
\end{equation}
The diffusion term of~\eqref{eq:Smoluchowskirho} is
\begin{multline} \label{eq:rhoRHS1}
    \int_{\Omega^2}\nabla_{\vx,\vy} \cdot \brac{\ind_{\Omega^2\setminus\Rset\cup\dRset}(\vx,\vy)\Dmat \nabla_{\vx,\vy}p(\vx,\vy,t)}v(\vx,\vy) \, d\vx \, d\vy \\
    = -\int_{\Omega^2\setminus\Rset\cup\dRset}\nabla_{\vx,\vy}v(\vx,\vy) \cdot\brac{\Dmat \nabla_{\vx,\vy}p(\vx,\vy,t)} \, d\vx \, d\vy.
\end{multline}
The association reaction term of~\eqref{eq:Smoluchowskirho} is
\begin{multline}
  \int_{\Omega^2}\alpha \delta_{\dRset}(\vx,\vy) \ind_{\Omega}(\gamma\vx+(1-\gamma)\vy) p(\vx,\vy,t) v(\vx,\vy) \, d\vx \, d\vy\\ 
  = \alpha\int_{\Omega^2\cap\dRset}\ind_{\Omega}(\gamma\vx+(1-\gamma)\vy) p(\vx,\vy,t) v(\vx,\vy) \, dS(\vx,\vy).
  \label{eq:rhoRHS2}
\end{multline}
Let $m_{\gamma}(\vx,\vy) = \gamma \vx + (1 - \gamma) \vy$ denote the position a newly produced $\textrm{C}$ particle is placed at. The dissociation term of \eqref{eq:Smoluchowskirho} is then
\begin{equation} \label{eq:rhoRHS3}
  \begin{aligned}
    \mum\int_{\Omega^2}&\brac{\int_{\Omega}\frac{\delta_{\partial B_{(1-\gamma)\rb}(\vO)}(\vx-\vz)}{\abs{\partial B_{(1-\gamma)\rb}(\vO)}}\delta\paren{\vy - \frac{\vz-\gamma\vx}{1-\gamma}}\pb(\vz,t)\,d\vz}\, v(\vx,\vy)\,d\vx \, d\vy \\
    &= \mum \tfrac{\paren{1-\gamma}^d}{{\abs{\partial B_{(1-\gamma)\rb}(\vO)}}} \int_{\Omega^2} \delta_{\partial B_{(1-\gamma)\rb}(\vO)}((1-\gamma)(\vx-\vy))\pb(m_{\gamma}(\vx,\vy),t) \\
    &\quad\quad\quad\quad\quad\quad\quad\quad\quad\quad\quad\quad\quad\quad\quad\quad \times \ind_{\Omega}\paren{m_{\gamma}(\vx,\vy)} v(\vx,\vy)\,d\vx \, d\vy \\
    &= \frac{\mum}{\abs{\partial B_{\rb}(\vO)}}\int_{\Omega^2}\delta_{\partial B_{\rb}(\vO)}(\vx-\vy)\pb\paren{m_{\gamma}(\vx,\vy),t} \ind_{\Omega}(m_{\gamma}(\vx,\vy)) v\paren{\vx,\vy}\,d\vx \, d\vy \\
    &= \frac{\mum}{\sqrt{2}\abs{\partial B_{\rb}(\vO)}}\int_{\Omega^2\cap\dRset}\pb\paren{m_{\gamma}(\vx,\vy),t}\ind_{\Omega}(m_{\gamma}(\vx,\vy))v\paren{\vx,\vy}\,dS(\vx,\vy). \\
  \end{aligned}  
\end{equation}
Here, in the second equality we have used the scaling properties of the surface $\delta$-function shown in the previous appendix, while in the third equality we have used the identity of Appendix~\ref{app:proofOfEqCR}.

Substituting~\eqref{eq:rhoLHS}, \eqref{eq:rhoRHS1}, \eqref{eq:rhoRHS2}, and~\eqref{eq:rhoRHS3} into~\eqref{eq:Smoluchowskirho}, we obtain the weak form of the equation for the probability density of the unbound state as
\begin{equation}
  \begin{aligned}
    \int_{\Omega^2\setminus\Rset\cup\dRset} &\PD{p}{t}(\vx,\vy,t) v(\vx,\vy) \, d\vx \, d\vy \\ 
    =&-\int_{\Omega^2\setminus\Rset\cup\dRset}\nabla_{\vx,\vy}v(\vx,\vy) \cdot\brac{\Dmat \nabla_{\vx,\vy}p(\vx,\vy,t)} \, d\vx \, d\vy \\ 
    &- \alpha\int_{\Omega^2\cap\dRset}\ind_{\Omega}(m_{\gamma}(\vx,\vy)) p(\vx,\vy,t) v(\vx,\vy) \, dS(\vx,\vy) \\ 
    &+ \frac{\mum}{\sqrt{2}\abs{\partial B_{\rb}(\vO)}}\int_{\Omega^2\cap\dRset}\pb(m_{\gamma}(\vx,\vy),t) 
    \ind_{\Omega}\paren{m_{\gamma}(\vx,\vy)}v(\vx,\vy)\, dS(\vx,\vy). 
  \end{aligned}
  \label{eq:WFrho}
\end{equation}

We now derive the corresponding weak form for $p(\vx,\vy,t)$ in the more common PDE with reactive boundary condition representation of the SCK CR model. Note, we abuse notation and again use $p(\vx,\vy,t)$ for the density, as we ultimately derive the same weak-form. The PDE version of the model is
\begin{align*}
    \PD{p}{t}(\vx,\vy,t) &= \nabla_{\vx,\vy} \cdot \brac{\Dmat \nabla_{\vx,\vy}p(\vx,\vy,t)}, \quad \forall (\vx,\vy) \in \Omega^2\setminus\Rset\cup\dRset, 
\end{align*}
with the reactive boundary condition
\begin{equation*}
    -\Dmat\nabla_{\vx,\vy}p(\vx,\vy,t)\cdot\veta(\vx,\vy) = 
    \ind_{\Omega}\paren{m_{\gamma}(\vx,\vy)} 
    \Big[\alpha p(\vx,\vy,t) - \tfrac{\mum}{\sqrt{2}\abs{\partial B_{\rb}(\vO)}} 
    \pb(m_{\gamma}(\vx,\vy),t) \Big],      
\end{equation*}
for $(\vx,\vy) \in \Omega^2 \cap\dRset$, and the reflecting zero Neumann boundary condition \eqref{eq:SCKCR_NoFluxBC}. We note that this representation may appear different than commonly used simplified forms, which are often written in the separation coordinate for the unbound state, see for example~\cite{AgmonSzaboRevRx1990}. There particles are assumed to move in free-space so that the two-particle dynamics in the unbound state can be formulated in terms of their scalar separation $r$ (assuming spherical symmetry). Our SCK model formulation allows all three particles to diffuse in a general domain, and so must track the probability densities to both be in a particular chemical state, and for each particle to be at specified positions. It likewise specifies the boundary condition at each point $(\vx,\vy)$ on the reactive boundary, as opposed to~\cite{AgmonSzaboRevRx1990}, where spherical symmetry allows specifying a boundary condition for the total flux into a reactive sphere (i.e. the probability per time of entering/leaving the bound state). The intrinsic bimolecular association rate constant of~\cite{AgmonSzaboRevRx1990}, $\kappa_a$, is related to our surface adsorption constant, $\alpha$, by
\begin{equation*}
  \kappa_a = \begin{cases}
    4 \sqrt{2} \pi \rb^2 \alpha, & \text{in three dimensions} \\
    2 \sqrt{2} \pi   \rb \alpha, & \text{in two dimensions.}
  \end{cases}
\end{equation*}
Our unbinding rate, $\mum$, should be identical to the intrinsic dissociation rate, $\kappa_d$ of~\cite{AgmonSzaboRevRx1990}.

For any test function $v(\vx, \vy) \in \mathcal{V}$, we have that the weak form of the SCK PDE representation is then
\begin{align*}
  \int_{\Omega^2\setminus\Rset\cup\dRset} \PD{p}{t}(\vx,\vy,t) &v(\vx,\vy) \, d\vx \, d\vy \\
  = &\phantom{-}\int_{\Omega^2\setminus\Rset\cup\dRset}\nabla_{\vx,\vy} \cdot \brac{\Dmat \nabla_{\vx,\vy}p(\vx,\vy,t)}v(\vx,\vy)\, d\vx \, d\vy, \\
    = &-\int_{\Omega^2\setminus\Rset\cup\dRset}\nabla_{\vx,\vy}v(\vx,\vy)\cdot\brac{\Dmat\nabla_{\vx,\vy}p(\vx,\vy,t)} \, d\vx \, d\vy \\
      \phantom{=} &+ \int_{\Omega^2\cap\dRset}v(\vx,\vy)\Dmat\nabla_{\vx,\vy}p(\vx,\vy,t)\cdot\veta(\vx,\vy)\, dS(\vx, \vy), \\
    = &-\int_{\Omega^2\setminus\Rset\cup\dRset}\nabla_{\vx,\vy}v(\vx,\vy)\cdot\Dmat\nabla_{\vx,\vy}p(\vx,\vy,t) \, d\vx \, d\vy \\
      \phantom{=} &- \alpha\int_{\Omega^2\cap\dRset}\brac{\ind_{\Omega}(m_{\gamma}(\vx,\vy))p(\vx,\vy,t)v(\vx,\vy)} \, dS(\vx,\vy) \\
      \phantom{=} &+ \tfrac{\mum}{\sqrt{2}\abs{\partial B_{\rb}(\vO)}}\int_{\Omega^2\cap\dRset}\pb(m_{\gamma}(\vx,\vy),t)\ind_{\Omega}\paren{m_{\gamma}(\vx,\vy)} v(\vx,\vy) \, dS(\vx,\vy), 
\end{align*}
which recovers \eqref{eq:WFrho}.

For the $\textrm{C} \to \textrm{A} + \textrm{B}$ unbinding reaction there are no reactive boundary conditions for $\pb(\vz,t)$, and hence the standard PDE form of the SCK CR model and our representation~\eqref{eq:bothDiffuseEqsrhob} are consistent for the dynamics of $\pb(\vz,t)$ (up to rewriting integration regions through evaluation of $\delta$-functions).

\section{Detailed Balance Leads to Rejection of Some Unbinding Events Near Boundaries for the Alternative VR Model}
\label{sect:DBrejectionAltVR}
\begin{theorem}
  \label{thm:plcmProbAltVR}
  Consider the detailed balance preserving unbinding kernel~\eqref{eq:bwdRxAltVR} for the alternative VR model~ \cite{Donevetal2018}. Denote the unbinding rate constant by $\mum$. For a C molecule at $\vz$ in $\Omega$
  \begin{equation} 
    \int_{\Omega^2}\km(\vx,\vy\vert\vz) \, d\vx \, d\vy < \mum
    \label{eq:AltVRPlcmProb}
  \end{equation}
  when $\vz$ is within $\rb$ of $\partial\Omega$. Similarly,
  \begin{equation} 
    \int_{\Omega^2}\km(\vx,\vy\vert\vz) \, d\vx \, d\vy = \mum
    \label{eq:AltVRPlcmProb_away_from_bd}
  \end{equation}
  if $\vz$ is at least $\rb$ away from $\partial\Omega$.
\end{theorem}
\begin{proof}
  For the unbinding kernel~\eqref{eq:bwdRxAltVR}, the probability per time a \textrm{C} molecule at $\vz$ will successfully produce \textrm{A} and \textrm{B} molecules each within $\Omega$ is given by
  \begin{equation}
    \begin{aligned}
      \km(\vz) &= \int_{\Omega^2}\km(\vx,\vy\vert\vz) \, d\vx \, d\vy,\\  
      &= \mum\brac{\rho\int_{\Omega}\frac{\ind_{B_{\rb}(\vz)}(\vy)}{\abs{B_{\rb}(\vO)}}\ind_{\Omega}\paren{\vz} \, d\vy + (1-\rho)\int_{\Omega}\frac{\ind_{B_{\rb}(\vz)}(\vx)}{\abs{B_{\rb}(\vO)}}\ind_{\Omega}\paren{\vz} \, d\vx},\\ 
      &= \mum\brac{\rho\frac{\abs{\Omega\cap B_{\rb}(\vz)}}{\abs{B_{\rb}(\vO)}} + (1-\rho)\frac{\abs{\Omega\cap B_{\rb}(\vz)}}{\abs{B_{\rb}(\vO)}}},\\
      &= \mum\frac{\abs{\Omega\cap B_{\rb}(\vz)}}{\abs{B_{\rb}(\vO)}}.
    \end{aligned}
    \label{eq:AltVRUnbindProbSuccess}
  \end{equation}
  Let $d(\vz,\partial\Omega)$ be the Euclidean distance from $\vz$ to $\partial\Omega$, and assume $d(\vz,\partial\Omega) < \rb$ so that 
  \begin{equation*}
    \abs{B_{\rb}(\vz)\cap\Omega} < \abs{B_{\rb}(\vO)}.
  \end{equation*}
  By \eqref{eq:AltVRUnbindProbSuccess} we immediately obtain
  \begin{equation*}
    \km(\vz) < \mum.
  \end{equation*}
  
  Finally, suppose $d(\vz,\partial\Omega)\geq\rb$ so that the ball of radius $\rb$ about $\vz$ is strictly contained in $\Omega$. In this case~\eqref{eq:AltVRUnbindProbSuccess} becomes
  \begin{equation}
    \km(\vz) = \mum\frac{\abs{B_{\rb}(\vz)}}{\abs{B_{\rb}(\vO)}} = \mum.
    \label{eq:AltVRUnbindSuccessXout}
  \end{equation}
\end{proof}

\section{Detailed Balance for the SCK CR Model}
\label{app:CRKd}
In the SCK CR model, the chosen reaction kernels are 
\begin{align}
    \kp(\vz \vert \vx, \vy) &= \alpha\sqrt{2}\delta_{\partial B_{\rb}(\vec{0})}(\vx - \vy)\delta(\vz - \gamma\vx - (1-\gamma)\vy), \label{eq:fwdRxSCK}\\
    \km(\vx,\vy \vert \vz) &= \mum\delta\paren{\vy - \frac{\vz-\gamma\vx}{1-\gamma}}\frac{\delta_{\partial B_{(1-\gamma)\rb}(\vO)}(\vx-\vz)}{\abs{\partial B_{(1-\gamma)\rb}(\vO)}}. \label{eq:bwdRxSCK}
\end{align}
To verify that these reaction terms satisfy detailed balance, we will show that \eqref{eq:detailedBalancekpkm} holds for all $(\vx, \vy) \in \OSqD$ and all $\vz \in \Omega$. Using \eqref{eq:detailedBalancekpkmints} the dissociation constant, $\Kd$, of the SCK CR model is given by
\begin{equation}
  \begin{aligned}
    \Kd &= \frac{\int_{\OSqD} \int_{\Omega} \mum\delta\paren{\vy - \frac{\vz-\gamma\vx}{1-\gamma}}\frac{\delta_{\partial B_{(1-\gamma)\rb}(\vO)}(\vx-\vz)}{\abs{\partial B_{(1-\gamma)\rb}(\vO)}} \, d\vz \, d\vx \, d\vy}
      {\int_{\OSqD} \int_{\Omega} \alpha\sqrt{2}\delta_{\partial B_{\rb}(\vec{0})}(\vx - \vy)\delta(\vz - \gamma\vx - (1-\gamma)\vy) \, d\vz \, d\vx \, d\vy}, \\
      &= \frac{\frac{\mum}{\abs{\partial B_{\rb}(\vO)}}\int_{\OSqD}\delta_{\partial B_{\rb}(\vO)}(\vx-\vy)\ind_{\Omega}\paren{\gamma\vx + (1-\gamma) \vy}\,d\vx \, d\vy}
      {\alpha\sqrt{2}\int_{\OSqD}\delta_{\partial B_{\rb}(\vO)}(\vx - \vy)\ind_{\Omega}(\gamma\vx + (1-\gamma)\vy) \, d\vx \, d\vy}, \\ 
      &= \frac{\mum}{\alpha\sqrt{2}\abs{\partial B_{\rb}(\vO)}}.
  \end{aligned}
  \label{eq:CR_Kd}
\end{equation}
Here we have used the scaling property that $\delta(\lambda\vx) = \delta(\vx)/\lambda^d$, the surface $\delta$-function scaling property~\eqref{eq:deltasurfmeasidentityAp2}, and that $\abs{\partial B_{(1-\gamma)\rb}(\vO)} = (1-\gamma)^{d-1}\abs{\partial B_{\rb}(\vO)}$. Reusing these properties again, with the proceeding choice for $\Kd$ we find
\begin{equation*}
  \begin{aligned}
    \Kd\kp(\vz \vert \vx,\vy) &= \frac{\mum}{\alpha\sqrt{2}\abs{\partial B_{\rb}(\vO)}}\alpha\sqrt{2}\delta_{\partial B_{\rb}(\vO)}(\vx - \vy)\delta(\vz - \gamma\vx - (1-\gamma)\vy), \\
    &= \mum\frac{\delta_{\partial B_{\rb}(\vO)}(\vx - \vy)}{\abs{\partial B_{\rb}(\vO)}}\delta(\vz - \gamma\vx - (1-\gamma)\vy), \\
    &= \mum\delta\paren{\vy - \frac{\vz-\gamma\vx}{1-\gamma}}\frac{\delta_{\partial B_{(1-\gamma)\rb}(\vO)}(\vx-\vz)}{\abs{\partial B_{(1-\gamma)\rb}(\vO)}}, \\
    &= \km(\vx,\vy \vert \vz).
  \end{aligned}
\end{equation*}

\section{Detailed Balance Leads to Rejection of Some Unbinding Events Near Boundaries for the SCK Model}
\label{sect:DBrejectionSCK}
\begin{theorem}
  \label{thm:plcmProbSCK}
  Consider the detailed balance preserving SCK unbinding kernel~\eqref{eq:mTwoPartDefCR}. Denote the unbinding rate constant by $\mum$. For a C molecule at $\vz$ in $\Omega$
  \begin{equation} 
    \int_{\Omega^2}\km(\vx,\vy\vert\vz) \, d\vx \, d\vy < \mum
    \label{eq:SCKPlcmProb}
  \end{equation}
  when $\vz$ is within $\max(\gamma,1-\gamma)\rb$ of $\partial\Omega$. Similarly,
  \begin{equation} 
    \int_{\Omega^2}\km(\vx,\vy\vert\vz) \, d\vx \, d\vy = \mum
    \label{eq:SCKPlcmProb_away_from_bd}
  \end{equation}
  if $\vz$ is at least $\max(\gamma,1-\gamma)\rb$ away from $\partial\Omega$.
\end{theorem}
\begin{proof}
  Given the unbinding rate $\mum$, the unbinding mechanism for the SCK CR model that satisfies detailed balance is given by \eqref{eq:mTwoPartDefCR}. We will again consider two cases: $\max(\gamma,1-\gamma) = 1-\gamma$ and $\max(\gamma,1-\gamma) = \gamma$. First assume $\max(\gamma,1-\gamma) = 1-\gamma$.
  
  The probability per time a \textrm{C} molecule at $\vz$ will successfully produce \textrm{A} and \textrm{B} molecules each within $\Omega$ is given by
  \begin{equation}
    \begin{aligned}
      \km(\vz) &= \int_{\Omega^2}\km(\vx,\vy\vert\vz) \, d\vx \, d\vy,\\  
      &= \mum\int_{\Omega}\frac{\delta_{\partial B_{(1-\gamma)\rb}(\vz)}(\vx)}{\abs{\partial B_{(1-\gamma)\rb}(\vO)}}\ind_{\Omega}\paren{\frac{\vz-\gamma\vx}{1-\gamma}} \, d\vx,\\ 
      &= \mum\frac{\abs{\partial B_{(1-\gamma)\rb}(\vz)\cap\hat{\Omega}_{\gamma}(\vz)}}{\abs{\partial B_{(1-\gamma)\rb}(\vO)}},
    \end{aligned}
    \label{eq:SCKUnbindProbSuccess}
  \end{equation}
  where the $\hat{\Omega}_{\gamma}(\vz)$ denotes the translated and dilated set given by \eqref{eq:omegahat}. Let $d(\vz,\partial\Omega)$ be the Euclidean distance from $\vz$ to $\partial\Omega$, and assume $d(\vz,\partial\Omega) < (1-\gamma)\rb$ so that 
  \begin{equation*}
    \abs{\partial B_{(1-\gamma)\rb}(\vz)\cap\Omega} < \abs{\partial B_{(1-\gamma)\rb}(\vO)}.
  \end{equation*}
  By \eqref{eq:SCKUnbindProbSuccess} this immediately gives
  \begin{equation*}
    \km(\vz) < \mum.
  \end{equation*}

  Now consider $\vz$ with $d(\vz,\partial\Omega)\geq(1-\gamma)\rb$ so that the ball of radius $(1-\gamma)\rb$ about $\vz$ is strictly contained in $\Omega$. Let $\vw \in \partial B_{(1-\gamma)\rb}(\vz) \subset \Omega$. We claim $\vw \in \hat{\Omega}_{\gamma}(\vz)$. Take
  \begin{equation*}
    \vv = \tfrac{1}{1 - \gamma} (\vz - \gamma \vw)
  \end{equation*}
  so that 
  \begin{equation*}
    \abs{\vv - \vz} = \frac{\gamma}{1-\gamma}\abs{\vz-\vw} = \gamma\rb.
  \end{equation*}
  This implies that $\vv \in \partial B_{\gamma \rb}(\vz)$, and therefore in $\Omega$. As such, $\vw \in \hat{\Omega}_{\gamma}(\vz)$ implying $\partial B_{(1-\gamma)\rb}(\vz) \subset \hat{\Omega}_{\gamma}(\vz)$. By~\eqref{eq:SCKUnbindProbSuccess}
  \begin{equation*}
    \km(\vz) = \mum\frac{\abs{\partial B_{(1-\gamma)\rb}(\vz)}}{\abs{\partial B_{(1-\gamma)\rb}(\vO)}} = \mum.
  \end{equation*}
  
Finally, using scaling properties of the $\delta$-function, we can equivalently write
\begin{align*}
  \km(\vx,\vy\vert\vz) &= \mum \tfrac{(1-\gamma)^d}{\gamma^d} \delta \paren{\vx - \frac{\vz - (1-\gamma)\vy}{\gamma}}
  \frac{\delta_{\partial B_{(1-\gamma)\rb}(\vO)}\paren{\tfrac{1-\gamma}{\gamma}(\vy-\vz)}}{\abs{\partial B_{(1-\gamma)\rb}(\vO)}} \\
  &= \mum \delta \paren{\vx - \frac{\vz - (1-\gamma)\vy}{\gamma}} \frac{\delta_{\partial B_{\gamma\rb}(\vO)}\paren{\vy-\vz}}{\abs{\partial B_{\gamma\rb}(\vO)}}. 
\end{align*}
Using the preceding formula, and integrating~\eqref{eq:SCKUnbindProbSuccess} first with respect to $\vx$, we obtain
  \begin{equation*}
    \begin{aligned}
      \km(\vz) &= \mum\int_{\Omega}\frac{\delta_{\partial B_{\gamma\rb}(\vz)}(\vy)}{\abs{\partial B_{\gamma\rb}(\vO)}}\ind_{\Omega}\paren{\frac{\vz-(1-\gamma)\vy}{\gamma}} \, d\vy,\\ 
      &= \mum\frac{\abs{\partial B_{\gamma\rb}(\vz)\cap\hat{\Omega}_{1-\gamma}(\vz)}}{\abs{\partial B_{\gamma\rb}(\vO)}}.
    \end{aligned}
  \end{equation*}
A similar argument to above interchanging $\gamma$ and $1-\gamma$ then gives the result when $\max(\gamma,1-\gamma) = \gamma$.
\end{proof}

\section{Equilibrium State for a Well-Mixed Stochastic Chemical Kinetics Model}
\label{sect:SSProbKd}
We consider a well-mixed stochastic chemical kinetics model for the reversible reaction
\begin{equation*}
  \textrm{A} + \textrm{B} \xrightleftharpoons[k_{-}]{k_{+}} \textrm{C},
\end{equation*}
based on the Chemical Master Equation~\cite{GardinerHANDBOOKSTOCH}. Here $k_{+}$ denotes the well-mixed association rate in units of volume per time, while $k_{-}$ is the dissociation rate in units of per time. Let $P(t)$ denote the probability the system is in the unbound state at time $t$, with $P_{\text{b}}(t)$ the probability the system is in the bound state. The well-mixed Chemical Master Equation model for the two-particle reversible reaction is
\begin{subequations} 
\label{eq:WellMixedModel}
  \begin{align} 
    \D{P}{t}(t) &= -\frac{k_+}{\abs{\Omega}} P(t) + k_- P_{\text{b}}(t), \label{eq:WellMixedP} \\    
    \D{P_{\text{b}}}{t}(t) &= -k_- P_{\text{b}}(t) + \frac{k_+}{\abs{\Omega}} P(t). \label{eq:WellMixedPb}
  \end{align}
\end{subequations}
Together with the normalization $P + P_{\text{b}} = 1$ we can eliminate \eqref{eq:WellMixedPb} and rewrite \eqref{eq:WellMixedP} as
\begin{equation*}
  \D{P}{t}(t) = -\frac{k_+}{\abs{\Omega}} P(t) + k_- (1 - P(t)),
\end{equation*}
which gives the steady-state solution as
\begin{equation}
  \equil{P} = \frac{k_-}{k_- + (k_+/\abs{\Omega})} = \frac{\Kd\abs{\Omega}}{1+\Kd\abs{\Omega}},
  \label{eq:WMSSP}
\end{equation}
where the dissociation constant, $\Kd$, is defined as
\begin{equation*}
  \Kd = \frac{k_-}{k_+}.
\end{equation*}
The steady-state solution to \eqref{eq:WellMixedPb} is similarly given by
\begin{equation}
  \equil{P}_{\text{b}} = 1-\equil{P} = \frac{1}{1+\Kd\abs{\Omega}}.
  \label{eq:WMSSPb}
\end{equation}

Provided $\Kd$ is chosen the same in the well-mixed and particle models, we note that~\eqref{eq:WMSSP} and~\eqref{eq:WMSSPb} give the same steady-state probabilities as we find for the particle model, see~\eqref{eq:ssprobs}.

\section{Brownian Dynamics Simulations} \label{sect:numericsappendix}
\begin{figure*}[tbph]
  \begin{algorithm}[H] 
    \caption{Brownian Dynamics (BD) method used in Section~\ref{sect:numerics}}
    \label{alg:BD}
    \begin{algorithmic}[1]
      \State{Input domain, $\Omega \subset \R^3$; diffusivity, $D$; on rate, $\lambda$; off rate, $\mum$; reaction radius, $\rb$; initial concentration of \textrm{C} molecules, $\brac{C_0}$; number of timesteps, $N_t$; and timestep, $\Delta t$.}
      \State{In the following \texttt{randn(3)} denotes a vector of three independent samples from the normal distribution with mean zero and variance one.}
      \State{Uniformly distribute $N_C = \abs{\Omega} \brac{C_0}$ molecules throughout $\Omega$. Set the number of \textrm{A} and \textrm{B} molecules to zero, $N_A=N_B=0$.}
      \State{Denote by $\vX_{i,S}(t)$ the position of the $i$th molecule of species $S \in \{A,B,C\}$.}
      \For{n in $1$ to $N_t$}
        \For{$S$ in $\{A,B,C\}$ and $i$ in $1$ to $N_S$}
          \State{$\vX_{i,S}(t+\Delta t) := \vX_{i,S}(t) + \sqrt{2 D \Delta t}$ \texttt{randn(3)}}
          \State{If $X_{i,s}(t + \Delta t) \not\in \Omega$, use normal reflection~\cite{SchussStochProcBook2010} to reflect it back into the domain.}
        \EndFor
        \For{all pairs $(\vX_{i,A},\vX_{j,B})$ where both molecules still exist}
          \State{If $\norm{\vX_{i,A} - \vX_{j,B}} < \rb$, react with probability $\lambda \Delta t$.}
          \If{reaction occurs}
          \State{Place a new \textrm{C} molecule at $\tfrac{1}{2}(\vX_{i,A}+\vX_{j,B})$ (i.e. $\gamma = \tfrac{1}{2}$ in~\ref{eq:bwdRxDoi}).}
          \EndIf
        \EndFor
        \For{all \textrm{C} molecules $\vX_{i,C}$}
          \State{Dissociate the molecule with probability $\mum \Delta t$.}
          \If{reaction occurs}
          \State{Place the \textrm{A} molecule uniformly in the sphere of radius $(1-\gamma)\rb$ about $\vX_{i,C}$.}
          \State{Place the \textrm{B} molecule such that $\vX_{i,C}$ is the midpoint between the \textrm{A} and \textrm{B} molecules.}
          \State{If any reaction product is outside $\Omega$ either:}
          \State{\hspace{30pt} (a) reject the reaction (detailed balance method).} 
          \State{\hspace{30pt} (b) reflect the product(s) back into the domain (reflection method).}
          \EndIf
        \EndFor
      \EndFor
    \end{algorithmic}
  \end{algorithm}
  \end{figure*}  
The Brownian Dynamics method used in Section~\ref{sect:numerics} is summarized in Algorithm~\ref{alg:BD} and based on the small timestep BD method for the Doi VR model described in~\cite{ErbanChapman2009}. It uses a standard first order Lie-Trotter splitting in time to first diffuse all particles over one timestep, then sample bimolecular reactions over one timestep based on particle positions after the diffusion step, and finally sample dissociation reactions over one timestep based on particle positions and numbers after the bimolecular reaction timestep. 

For $L = 20 \text{ nm}$ and $\Omega = (0,L)^3$ we used $\lambda = 40.5745 \text{ s}^{-1}$, giving a steady-state probability to be in the \textrm{C} state of $\bar{P}_{\text{b}} = .5512$ for a system with one initial \textrm{C} particle. Values of $\lambda$ for other domain sizes were calculated by solving the steady-state well-mixed chemical master equation~\cite{McQuarrieJAppProb} for the number of \textrm{C} particles in the system numerically, and using Matlab's \texttt{fzero} command to optimize these solutions to recover the desired steady-state fraction $\lim_{t \to \infty} f(t)=.551$, see~\eqref{eq:Cfrac}, to three digits of accuracy. The association rate in the well-mixed chemical master equation was taken to be $\lambda \tfrac{4}{3} \pi \rb^3$ and the dissociation rate was $\mum$, see Table~\ref{tab:tab1} for numerical values of $\mum$ and $\rb$. For a cubic domain with $L = 100 \text{ nm}$ we found $\lambda = 89.7914 \text{ s}^{-1}$, and for a cubic domain with $L = 200 \text{ nm}$ we found $\lambda = 90.3269519 \text{ s}^{-1}$. For a domain of dimensions $L \times L \times 30 \text{nm}$ with $L = 200 \text{ nm}$ we found $\lambda = 89.8725 \text{ s}^{-1}$.  

All BD simulations used a timestep of $\Delta t = 10^{-8}\textrm{ s}$. With $D=1 (\mu \textrm{m})^2 / \textrm{s}$ as in our simulations, this timestep gives a standard deviation in spatial displacements due to diffusion per timestep of $\sqrt{6 D \Delta t} \approx .24 \textrm{ nm}$, substantially smaller than the reaction-radius of $10 \textrm{ nm}$. With the dissociation rate of $17.3 \textrm{ s}^{-1}$, it gives a probability per timestep that a \textrm{C} particle dissociates of $\mum \Delta t = 1.73 \times 10^{-7}$. For a pair of \textrm{A} and \textrm{B} particles within $\rb$ it gives a probability per timestep of reacting of $\lambda \Delta t \in \brac{4 \times 10^{-7},9.1 \times 10^{-7}}$ as $\lambda$ was varied between domain sizes.

\clearpage 

\bibliographystyle{aipnum4-1.bst}
\bibliography{lib.bib}

\end{document}